\documentclass[12pt]{article}

\usepackage[font={small,sl}]{caption}
\let\chapter\section
\usepackage[boxed]{algorithm2e}
\usepackage{amsmath}
\usepackage{bbm}
\usepackage{booktabs}
\usepackage{natbib}
\setcitestyle{authoryear,open={(},close={)}}
\usepackage{tipa}
\renewcommand{\gamma}{\text{\textbabygamma}}
\usepackage{amssymb}
\usepackage{amsthm}
\newtheorem{theorem}{Theorem}
\newtheorem{lemma}{Lemma}
\theoremstyle{plain}

\theoremstyle{definition}
\newtheorem{definition}{Definition}
\usepackage{amsbsy}
\usepackage{amsfonts}
 \usepackage{tikz}
 \usetikzlibrary{arrows,snakes}
\let\abs=\envert

\newcommand{\norm}[1]{\lVert#1\rVert}
\newcommand{\ind}[1]{\mathbbm{1}{\{#1\}}}
\newcommand{\bone}[1]{\mathbf{1}_{#1}}
\newcommand{\given}{\,|\,}

\DeclareMathOperator{\var}{Var}

\DeclareMathOperator{\argmin}{arg\,min}

\usepackage{amsmath}
\usepackage{bbm}
\usepackage{booktabs}
\usepackage{amssymb}
\usepackage{tipa}
\renewcommand{\gamma}{\text{\textbabygamma}}

 \usepackage{tikz}
 \usetikzlibrary{arrows,snakes}
\title{Likelihood Ratio Tests for a Dose-Response Effect using Multiple Nonlinear Regression Models}

\author
{Georg Gutjahr \\
Department of Mathematics, University of Bremen, Germany\\
\texttt{georg.gutjahr@math.uni-bremen.de}
\and
Bj\"{o}rn Bornkamp \\
Biostatistical Sciences and Pharmacometrics, Novartis Pharma
AG, Basel, Switzerland\\
\texttt{bjoern.bornkamp@novartis.com}}

\date{}

\begin{document}

\maketitle

\begin{abstract}

  We consider the problem of testing for a dose-related effect based
  on a candidate set of (typically nonlinear) dose-response models
  using likelihood-ratio tests. For the considered models this reduces
  to assessing whether the slope parameter in these nonlinear
  regression models is zero or not. A technical problem is that the
  null distribution (when the slope is zero) depends on
  non-identifiable parameters, so that standard asymptotic results on
  the distribution of the likelihood-ratio test no longer
  apply. Asymptotic solutions for this problem have been extensively
  discussed in the literature. The resulting approximations however
  are not of simple form and require simulation to calculate the
  asymptotic distribution. In addition their appropriateness might be
  doubtful for the case of a small sample size.  Direct simulation to
  approximate the null distribution is numerically unstable due to the
  non identifiability of some parameters. In this article we derive a
  numerical algorithm to approximate the exact distribution of the
  likelihood-ratio test under multiple models for normally distributed
  data. The algorithm uses methods from differential geometry and can
  be used to evaluate the distribution under the null hypothesis, but
  also allows for power and sample size calculations.  We compare the
  proposed testing approach to the MCP-Mod methodology and alternative
  methods for testing for a dose-related trend in a dose-finding
  example data set and simulations.

% Consider a set of nonlinear regression models for the
% mean vector of normally distributed
% observations and the hypothesis
% that at least one of these models fits the
% data better than a constant model.
% For a single ``sufficiently smooth'' model,
% Hotelling showed that the
% likelihood-ratio test statistic is a monotonous function
% of the correlation between the
% observations and the maximum-likelihood
% prediction; using methods from differential geometry,
% the exact null distributions of this statistic
% can be obtained.
% For multiple models,
% the likelihood-ratio statistic
% is the minimum of
% the individual likelihood-ratio statistics.
% The null distribution is determined by volumes
% of tubular neighborhoods on the unit sphere.
% We describe how such volumes can be approximated numerically.
% This approach can also be used to calculate
% the distribution under alternative hypotheses
% and it does not require that the models
% are smooth.
% We apply the likelihood-ratio test to
% data from a dose-response clinical trial
% and compare the its performance
% against several other trend-test procedures
% in some scenarios motivated by this example data.

\end{abstract}

\section{Introduction}
\label{sec:intro}

A major objective in the development of a pharmaceutical compound is
the characterisation of its dose-response curve. For this purpose
Phase II trials are conducted that compare several doses of the
compound to placebo. Then (typically) nonlinear regression models are
used to estimate the underlying dose-response curve. See for example
\cite{bret:pinh:bran:2005, thom:2006, drag:hsua:padm:2007,
  jone:layt:rich:2011} or \cite{grie:kram:2005} for different
approaches towards model-based dose-response analyses in Phase II
studies.

Typically the model linking dose and response can be assumed to follow
a nonlinear function of the form $\alpha+\beta x_{\gamma}$, where
$\alpha$ and $\beta$ are linear parameters describing the placebo
response and the slope parameter, and $x_{\gamma}$ is a nonlinear
transformation of the dose variable $z$ depending on a parameter
vector $\gamma$. The major question regarding the dose-response curve
that we will consider in this article is to assess whether there
exists a dose-related effect (\textit{i.e.} a dose-response trend) or
not. For the dose-response function above this reduces to testing the
hypothesis of $\beta=0$.

One of the challenges with a dose-response model based approach is
that there is model uncertainty at the design stage of the trial, when
the statistical analyses are specified. That means specifying one
particular form of the dose-response curve bears the risk of
mis-specification. Naively one might think that it is valid to apply a
model selection procedure to obtain the best model once one has
obtained the data and then perform a test for a dose-related effect,
ignoring the fact that a model selection was performed. However it is
known that statistical inference following a model selection is no
longer valid in the sense of not providing confidence intervals of
nominal coverage and resulting in a type I error inflation
(\cite{chat:1995, leeb:poet:2005}, or Chapter 7 of
\cite{clae:hjor:2008}). These concerns might have lead to a situation
where primarily ANOVA type methods have been used in Phase II trials,
which do not assume a functional relationship between dose and
response. This, however, comes at a loss in terms of statistical
efficiency but also interpretability of the study results.

A compromise between approaches that make no assumptions on the
functional form of the dose-response curve and those that assume one
specific dose-response model is to specify a candidate set of
dose-response models. This is the idea underlying the Multiple
Comparison Procedures and Modeling (MCP-Mod) approach
(\cite{bret:pinh:bran:2005, pinh:born:glim:2014,ema:2014}).  MCP-Mod
consists of two steps: testing and estimation.  The testing step,
which is of primary interest for this article, is done using multiple
linear contrast tests. To derive contrast tests that are powerful to
detect the nonlinear shapes in the candidate set, one needs to
pre-specify the parameters $\gamma$ of the nonlinear part of the
regression functions.  This approach bears the risk of model
mis-specification (as a $\gamma$ different from the one pre-specified
might be adequate).

An alternative approach is to assess the hypothesis $\beta=0$ by using
multiple likelihood-ratio tests. Arguments by \citet{andrews:1995} and
\citet{andrews:1996} show that the likelihood-ratio test will always
be admissible in this setting (\textit{i.e.}, there is no test that
controls the type-I error and is uniformly more powerful).  A
technical problem with such likelihood-ratio tests is that under the
null hypothesis the parameters $\gamma$ of the nonlinear regression
models are asymptotically not identifiable.  Therefore, the standard
asymptotic results for the null distribution of likelihood-ratio
statistics do no longer apply; see for example \citet{davi:1977},
\citet{davi:1987}, \citet{andr:plob:1994}, \citet{ritz:skov:2005} and
\citet{liu:shao:2003}.  \cite{dett:tito:bret:2015} derive the
asymptotic null distribution of the likelihood-ratio statistic under
multiple models.  \cite{baay:houg:2015} only consider nested
dose-response models, but describe a similar approach, where the
asymptotic distribution for the likelihood-ratio tests is derived
explicitly.  The resulting asymptotic distribution is not of a simple
form and requires simulation to approximate critical values and
p-values that are asymptotically valid.

We will consider the exact distribution for likelihood-ratio tests for
the null hypothesis that no trend exists against the composite
alternative that one of the candidate models is true.  We assume that
observations are independent and normally distributed and that the
dose-response models are piece-wise continuous, which is generally true
for the models considered in practice.

To find the small-sample distribution of the likelihood-ratio
statistic, one can of course simulate data and fit the dose-response
models under the null hypothesis ($\beta=0$) to obtain critical values
for the test statistic.  However, as most of the statistical models
considered in this article are nonlinear, computationally expensive
iterative techniques are required to calculate the maximum likelihood
estimates for each model and each simulated data set.  Such an
approach is also numerically unstable under $\beta=0$ as the
likelihood function will be almost flat with potentially several local
maxima for the non-identifiable parameter $\gamma$.

In this article, we work with the small-sample distribution of the
likelihood-ratio statistic using methods developed by \cite{hote:1939}
and \cite{weyl:1939} and reviewed from a statistical perspective by
\cite{joha:john:1990}.  We start with the special case of a single
model.  For this special case, Hotelling used a geometric approach to
derive the exact null distributions analytically.  We extend these
methods to multiple dose-response models and show that, as in the case
of a single model, the distribution is determined by volumes of
certain tubular neighborhoods on the unit sphere.  We then present an
importance sampling-type algorithm to approximate such volumes
numerically.  In particular, this approach does not require the
calculation of maximum-likelihood estimates (with its associated
numerical difficulties) in each simulation run. Using this approach,
it is also possible to evaluate the power of the proposed test under
alternative hypotheses, thereby enabling sample-size calculations,
which are crucial in clinical development.

The outline of this paper is as follows.  After introducing the
notation in Section~\ref{sec:not}, we will review the application of
Hotelling's approach for the case of one dose-response model in
Section~\ref{sec:onemod}.  We then present the new methodology in
Section~\ref{sec:multmod} and Section~\ref{sec:num}, where a numerical
algorithm is introduced to approximate the distribution of the
likelihood-ratio test.  In Section~\ref{sec:appl} we apply the method
to data from a dose-response study and compare the performance of the
new method to MCP-Mod and a number of alternative approaches in a
setting motivated by this example data.  Section~\ref{sec:concl} gives
some concluding remarks.

\section{Methods}
\label{sec:methods}

\subsection{Notation}
\label{sec:not}

Consider a random vector $y$ containing the clinical measurement of
interest for each of the $n$ patients. Here we assume $y \sim
\mathcal{N}(\mu, \sigma^2 I_n)$ with $n$ independent, normally
distributed observations with an unknown common standard deviation
$\sigma \in \mathbb{R}_+$ and an unknown mean vector $\mu \in
\mathbb{R}^n$.

To describe possible forms of the mean vector $\mu$, one selects $m$
candidate dose-response models of the following partially-linear form
\begin{equation}
  \label{eq:mu}
  \mu_i(\alpha, \beta, \gamma) = \alpha \bone{n}
  +\beta\,x_{\gamma, i}
\quad (\alpha, \beta\in\mathbb{R},\;\gamma\in\Gamma_i,\;i = 1,
\dots, m)
\end{equation}
%Here the $\alpha_m$'s and $\beta_m$'s
%are real parameters, the $\tau_m$'s
%are nonlinear transformations of the dose levels $\lambda$,
%and $\mathbf{1}_N$ is
%the vector of $N$ ones.
%where % the $\Gamma_i$'s are compact sets
%and the transformations
where $x_{\gamma,i} = x_{\gamma,i}(z)$ are nonlinear transformations
of the dose variable $z$.  Since the value of $z$ can be fixed
throughout, the dependencies on $z$ will not be explicitly indicated.
We assume that the nonlinear transformations $\gamma\mapsto
x_{\gamma,i} : \Gamma_i \to \mathbb{R}^n$ are piece-wise-continuous
functions. % on $\sigma$-compact sets $\Gamma_i$.
One example is the so-called Emax model
\[
\mu_i(\alpha,\beta,\gamma) = \alpha + \beta x_{\gamma},
\qquad x_{\gamma}= z/(z + \gamma),
\qquad \alpha,\beta \in \mathbb{R}, \; \gamma \in \Gamma = \mathbb{R}_+;
\]
for more examples, see, e.g., ~\cite{born:pinh:bret:2009}.

Note that for $\beta = 0$, the value of $\gamma$ in \eqref{eq:mu} has
no influence on the shape of the model function, which means that
asymptotically the parameter is not identified.
% ; even for $\beta \neq 0$, we do not require identifiability
% of the parameters ---that is, we may have $\mu_i(\alpha, \beta,
% \gamma_1) = \mu_i(\alpha, \beta, \gamma_2)$ for $\gamma_1 \neq
% \gamma_2$---although identifiability should be imposed if we are
% interested in parameter estimation in addition to hypothesis testing.

We will consider the hypothesis $\mu \in H_0$, where
\begin{equation}
  \label{eq:nullhypothesis}
  H_0 = \{\alpha \bone{n} : \alpha \in \mathbb{R}\},
\end{equation}
the one-sided alternatives $\mu \in H_1, \dots, \mu\in H_{2m}$,
where, for $i = 1, \dots, m$,
\begin{equation}
  \label{eq:alternativeplus}
  H_i = \{\mu_i(\alpha, \beta, \gamma) :
  \alpha \in \mathbb{R}, \beta \in \mathbb{R}_+, \gamma \in \Gamma_i\},
\end{equation}
\begin{equation}
  \label{eq:alternativeminus}
  H_{m + i} = \{\mu_i(\alpha, -\beta, \gamma) :
  \alpha \in \mathbb{R}, \beta \in \mathbb{R}_+, \gamma \in \Gamma_i\},
\end{equation}
and the multiple alternatives $\mu \in H_{I}$, where
\begin{equation}
  \label{eq:alternatives}
  H_{I} = \bigcup_{i\in I} H_i,
  \quad\text{for}~I \subseteq \{1, \dots, 2 m\}.
\end{equation}
For the special cases $I = \{1, \dots, m\}$ or $I = \{m + 1, \dots,
2m\}$ we get one-sided alternatives, and for the special case $I =
\{1, \dots, 2 m\}$ a two-sided alternative.  The goal will be to
derive the likelihood-ratio test for $H_0$ against $H_I$.

\subsection{Distribution of the LR test}
\label{sec:distribution-lr-test}

To illustrate the general approach for calculation of the distribution
of the likelihood-ratio test we will start with the simple setting of
a single linear dose-response model first in Section
\ref{sec:single-linear-dose}, then consider a single nonlinear model
in Section \ref{sec:onemod} and finally present the situation for
multiple dose-response models in Section \ref{sec:multmod}.

\subsubsection{Single linear dose-response model}
\label{sec:single-linear-dose}

We first consider the special case of testing $H_0$ against a single
hypothesis $H_i$ (so that $I = \{i\}$), where $\Gamma_i =
\{\gamma_i\}$ contains a single parameter value $\gamma_i$. In this
case $\gamma_i$ is known and the dose-response model reduces to a linear
regression model.  If we define $x = x_{\gamma_i, i}$, then $H_i$
becomes the hypothesis $\beta > 0$ in the linear model $y \sim
\mathcal{N}(\alpha\bone{n} + \beta x,\; \sigma^2 I_n)$, while $H_0$
states that $\beta = 0$. In this case standard distributional results
can be used to calculate the distribution of the LR test statistic for
testing $\beta=0$. The reason, why we present this case here is to
motivate a transformation of the standard LR test statistic that turns
out to be useful for subsequent sections.

Appendix A shows that the LR statistic has
the form $S(R) = (1 - \ind{R > 0} R^2)^{n/2}$ with
\begin{equation}
  \label{eq:stat}
  R = \tilde{x}^{\top} \tilde{y},\qquad\text{where}\quad
  \tilde{x} = \frac{B x}{\norm{Bx}} \quad\text{and}\quad
  \tilde{y} = \frac{B y} {\norm{By}},
\end{equation}
where the rows of the
$(n -1) \times n$ matrix $B$
form an orthonormal basis for
the linear subspace
$\mathbb{L} = \{a \in \mathbb{R}^{n} :
a^{\top} 1_{n} = 0\}$ and $\norm{\cdot}$ denotes the Euclidean norm.
%and $C = I_n - n^{-1}\bone{n}\bone{n}^{\top}$
%denotes the centering matrix.
The LR test rejects $H_0$
if $S$ is small enough, or equivalently,
since $S$ is non-increasing in $R$,
if $R$ is large enough.

The model predicts that the mean vector
has the form $\alpha \bone{n} + \beta x$, so that the centered and
scaled prediction from the model equals
$B(\alpha 1_n + \beta x) / \norm{B(\alpha 1_n + \beta x)} =
(B x) / \norm{B x} = \tilde{x}$ in the basis $B$.
Furthermore, $\tilde{y}$ contains
the standardized (centered and scaled)
observations, so that both $\tilde{x}$ and $\tilde{y}$ lie on the unit
sphere.  The centering and scaling implies that $R$ is the correlation coefficient between
observations and predictions and $R^2$ the coefficient of
determination \citep[][Section 4.4]{seber:2003}.
Since the inner product between the two unit
vectors $\tilde{x}$ and $\tilde{y}$ increases
monotonically with the distance
between the two vectors, the LR test rejects $H_0$
if the standardized predictions are close enough
to the standardized observations.

Let $d = n - 2$ denote the degrees of freedoms
and $\mathbb{S} = \{s \in \mathbb{R}^{d+1} : \norm{s} = 1\}$
the $d$-dimensional unit sphere
(some authors call it the $d+1$ dimensional unit sphere
since it is in $\mathbb{R}^{d+1}$).
Then $\tilde{x} \in \mathbb{S}$ and
$\tilde{y}$ is
uniformly distributed on $\mathbb{S}$
under $H_0$,
which can be seen as following:
the assumption of $H_0$ is
$y \sim \mathcal{N}(\alpha \bone{d + 2}, \sigma^2 I_{d + 2})$;
hence $By \sim \mathcal{N}(\alpha B \bone{d + 2}, \sigma^2 B B^{\top})$
with $\alpha B \bone{d + 2} = \mathbf{0}_{d + 1}$
(since $\bone{d + 2}$ is orthogonal to $\mathbb{L}$)
and $B B^{\top} = I_{d + 1}$;
hence $B y$ is spherically symmetric~\citep{fang:1990} and
$\tilde{y}$ is uniformly distributed on $\mathbb{S}$.

% Then under $H_0$ the standardized observations
% $\tilde{y}$ is uniformly distributed
% on a $n - 2$ dimensional sphere:
% First, $\tilde{y} \in \mathbb{S}^n \cap \mathbb{L}^n$
% with
% $\mathbb{S}^n = \{s \in \mathbb{R}^{n} : \norm{s} = 1\}$
% the $n - 1$ dimensional unit sphere
% and $\mathbb{L}^n = \{a \in \mathbb{R}^{n}:
% a^{\top}\mathbf{1}_{n} = 0\}$
% the $n - 1$ dimensional linear space orthogonal to
% the $\mathbf{1}$ vector.
% Let the rows of the $(n -1) \times n$ matrix
% $B$ form an orthonormal basis for $\mathbb{L}^n$.
% From $\tilde{x}, \tilde{y} \in \mathbb{L}^d$
% follows $\tilde{x}^{\top} \tilde{y} =
% (B \tilde{x})^{\top}(B \tilde{y})$ and since
% $B C y$ is a vector of iid normal random variables,
% it follows that $B \tilde{y}$ is uniformly distributed on
% $\mathbb{S}^{n-1}$ under $H_0$.

The p-value of the LR test is $P_0(R > r)$,
with $r$ for the observed value of $R$
and $P_0(\cdot)$ the probability
calculated under the null hypothesis $H_0$
(since the distribution of $R$ does not
depend on the parameter $\alpha$, we
may assume $\mu = \mathbf{0}_n$ for $P_0$).
If we define the neighborhood around a point $s \in \mathbb{S}$ (a
spherical cap) as
$\mathbb{C}_{s r} = \{t \in \mathbb{S} : s^{\top} t >r\}$,
then $P_0(R > r)$ is equal to the probability that $\tilde{y}$ is in the
spherical cap around $\tilde{x}$ defined by $r$ under $H_0$: $P_0(R > r)=P_0(\tilde{y}
\in \mathbb{C}_{\tilde{x} r})$; see Figure~\ref{fig:cap}.
As $\tilde{y}$ is uniformly distributed on $\mathbb{S}$
under $H_0$ it follows that $P_0(\tilde{y} \in \mathbb{C}_{\tilde{x} r}) = \abs{\mathbb{C}_{\tilde{x}r}} /
\abs{\mathbb{S}}$ where $\abs{\cdot}$ is the
($d$-dimensional) volume of a set on $\mathbb{S}$.
Since the volume of the spherical cap
$\abs{\mathbb{C}_{s r}}$ does not depend on
the point $s \in \mathbb{S}$,
we will simply write $\abs{\mathbb{C}_{r}}$
for such a volume.

Using the equations for the volumes of spheres
and of spherical caps \citep{li:2011},
we get
\begin{equation}
  \label{eq:cr}
  P_0(R > r) = \frac{\abs{\mathbb{C}_{r}}} {\abs{\mathbb{S}}} =
  \frac{1 - F(r^2, 1/2, d / 2)}{2},\qquad\text{for $r \in[0,1]$},
\end{equation}
where $F(\cdot, a, b)$ denotes the \textsc{cdf} of the
beta distribution with parameters $a, b \in\mathbb{R}_+$.
The equation when $r \in [-1,0)$ follows
from $\abs{\mathbb{C}_{r}} = \abs{\mathbb{S}} - \abs{\mathbb{C}_{-r}}$.

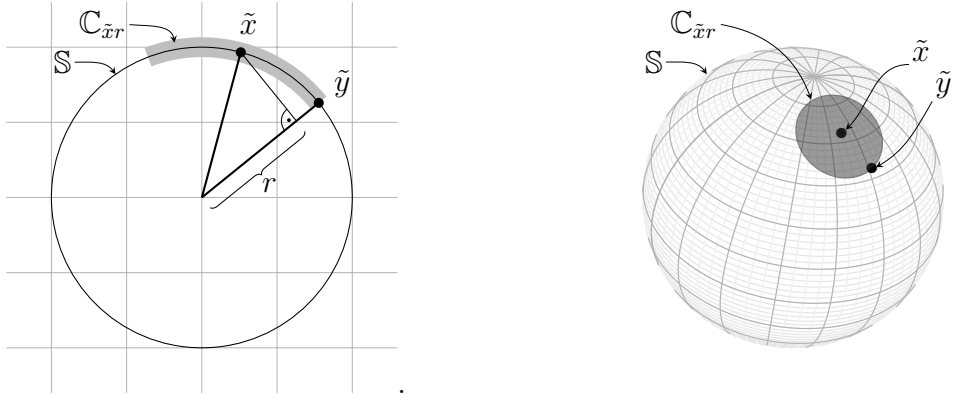
\begin{figure}
%  \hfill\raisebox{0.3cm}{\includegraphics{fig1}}%
%  \hfill\includegraphics{fig2}\hfill
 \hfill
% ##### Calculations
% type1 <- 0.2
% xt <- stat(getB(3), 1:3)
% theta <- atan2(xt[2], xt[1])
% 360 * theta / (2 * pi)
% r <- linear_crit(type1, 1)
% acos(r) / (2 * pi) * 360
% r
\begin{tikzpicture}[scale = 2,>=stealth]
  \def\xtAngle{75}
  \def\rAngle{36}
  \def\rLength{0.809017}

  \begin{scope}
    \draw[style=ultra thin, gray!60!white,step=0.5cm] (-1.3,-1.3) grid (1.3,1.3);
    % \foreach \x/\xtext in {-.5/-\frac{1}{2}}
    %     \draw[xshift=\x cm] (0pt,1pt) -- (0pt,-1pt)
    %         node[below,fill=white] {$\xtext$};
    % \foreach \y/\ytext in {-.5/-\frac{1}{2}}
    %     \draw[yshift=\y cm] (1pt,0pt) -- (-1pt,0pt)
    %         node[left,fill=white] {$\ytext$};
    %\draw (-1pt,0) -- (1pt,0);
    %\draw (0,-1pt) -- (0,1pt);
    %\draw (-1pt, -1pt) node[below left,fill=white] {$0$};
  \end{scope}

  \draw [gray!50!white,line width=7.5] (\xtAngle - \rAngle:1) arc (\xtAngle - \rAngle:\xtAngle + \rAngle:1);
  \draw [thin] (0,0) circle (1);
  \draw [thick] (0,0) -- (\xtAngle:1) [fill=black] circle(0.75pt);
  \node at (\xtAngle:1.2) {$\tilde{x}$};
  \node at (\xtAngle - \rAngle:1.2) {$\tilde{y}$};

  \draw [thick] (0,0) -- (\xtAngle - \rAngle:1) [fill=black] circle(0.75pt);
  \draw (\xtAngle:1) -- (\xtAngle - \rAngle:\rLength);

  \draw[snake=brace,raise snake=5pt,mirror snake]
  (0,0) -- (\xtAngle - \rAngle:\rLength);
  \node at (0.45, 0.08) {$r$};
  \node [inner sep = 1pt] (Cxr) at (-0.66, 1.15) {$\mathbb{C}_{\tilde{x}r}$};
  \draw [very thin,->] (Cxr.east) to [out=355,in=100] (100:1.045);
  \node [inner sep = 1pt] (S) at (-0.925,0.9) {$\mathbb{S}$};
  \draw [very thin,->] (S.east) to  [out=0,in=125] (125:1);

  \begin{scope}[shift=(\xtAngle - \rAngle:\rLength)]
    \draw (\xtAngle - \rAngle + 90:0.1)
    arc (\xtAngle - \rAngle + 90: \xtAngle - \rAngle + 180:0.1);
    \node at (\xtAngle - \rAngle + 135:0.05)
    [circle,draw=black,fill=black,inner sep=0.2] {};
  \end{scope}
\end{tikzpicture}.
\hfill
 \begin{tikzpicture}[scale = 2,,>=stealth]
  \draw[style=help lines,step=0.5cm,white] (-1.3,-1.3) grid (1.3,1.3);
  \pgftext {\includegraphics[scale=0.29]{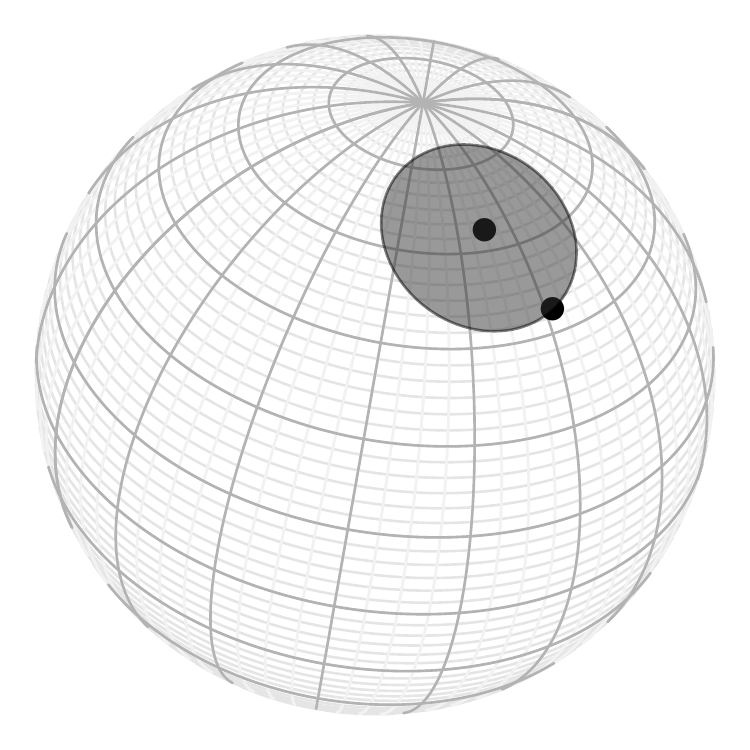}};
  \node [inner sep = 1pt] (Cxr) at (-0.66, 1.15) {$\mathbb{C}_{\tilde{x}r}$};
  \draw [->,very thin] (Cxr.east) to [out=350,in=110] (0.12, 0.65);
  \node [inner sep = 1pt] (S) at (-0.925,0.9) {$\mathbb{S}$};
  \draw [->,very thin] (S.east) to  [out=0,in=125] (125:1);
  \node [inner sep = 1pt] (xt) at (0.85, 0.975) {$\tilde{x}$};
  \draw [->,very thin] (xt.south west) to  [out=180,in=45] (0.35, 0.45);
  \node [inner sep = 1pt] (yt) at (1, 0.75) {$\tilde{y}$};
  \draw [->,very thin] (yt.south west) to  [out=240,in=30] (0.55, 0.2);
\end{tikzpicture}\hspace*{\fill}
  \caption{Standardized prediction $\tilde{x}$,
    standardized observation $\tilde{y}$, and
    spherical cap $\mathbb{C}_{\tilde{x}, r}$ (gray area)
    when $d = 1$ (left side) and when
    $d = 2$.}
  \label{fig:cap}
\end{figure}

Of course, in this section we have only derived a test equivalent to
the standard t-test for the linear model. However, the considerations
will be useful in the following sections.

\subsubsection{Single nonlinear dose-response model}
\label{sec:onemod}

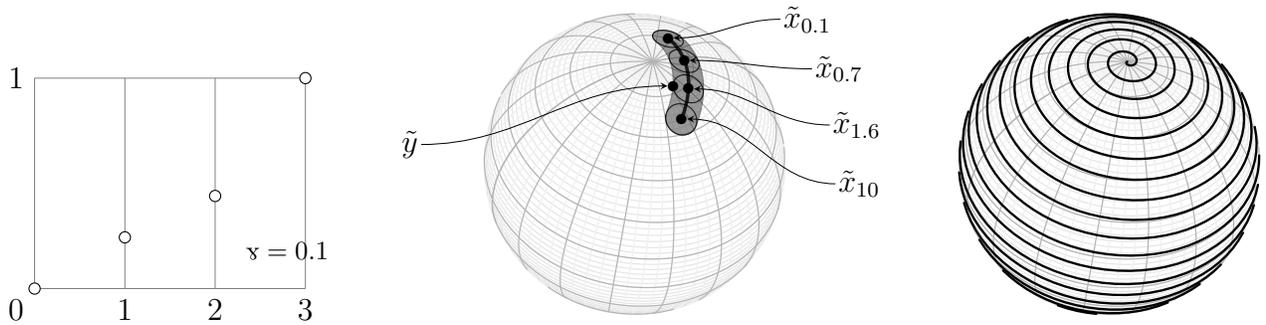
\begin{figure}
\begin{tikzpicture}[yscale=2.8,xscale=1.2,
  thinline/.style={color=gray,very thin}]
  \def\theta{0.4}
  \draw (0,0) node [anchor=north east, color=black] {0};
  \draw [thinline] (0,0)--(3,0);
  \draw [thinline] (0,0)--(0,1) node [anchor=east, color=black] {1};
  \draw [thinline] (0,1)--(3,1);
  \foreach \k in {1, ..., 3}
      \draw [thinline] (\k,1) -- (\k,0) node [anchor=north, color=black] {\k};
    \foreach \k in {4,...,1}
        \draw [gray!50!black, smooth, thick]
               plot file {spheretube\k.table};
  \foreach \x/\y in {0/0, 1/0.2434817, 2/0.4401643, 3/1}
      \node [inner sep = 1.5pt,draw,circle,fill=white] at (\x, \y) {};
  \node at (2.8, 0.18) {\footnotesize $\gamma = 0.1$};
%  \node at (1.7, 0.1) {\tiny $\gamma = 0.7$};
%  \node at (1.5, 0.25) {\tiny $\gamma = 1.6$};
%  \node at (1.5, 0.45) {\tiny $\gamma = 10$};
\end{tikzpicture}
 \hfill
 \begin{tikzpicture}[scale = 2,>=stealth]
   % \draw[style=help lines,step=0.5cm,white] (-1.3,-1.3) grid (1.3,1.3);
   \pgftext {\includegraphics[scale=0.29]{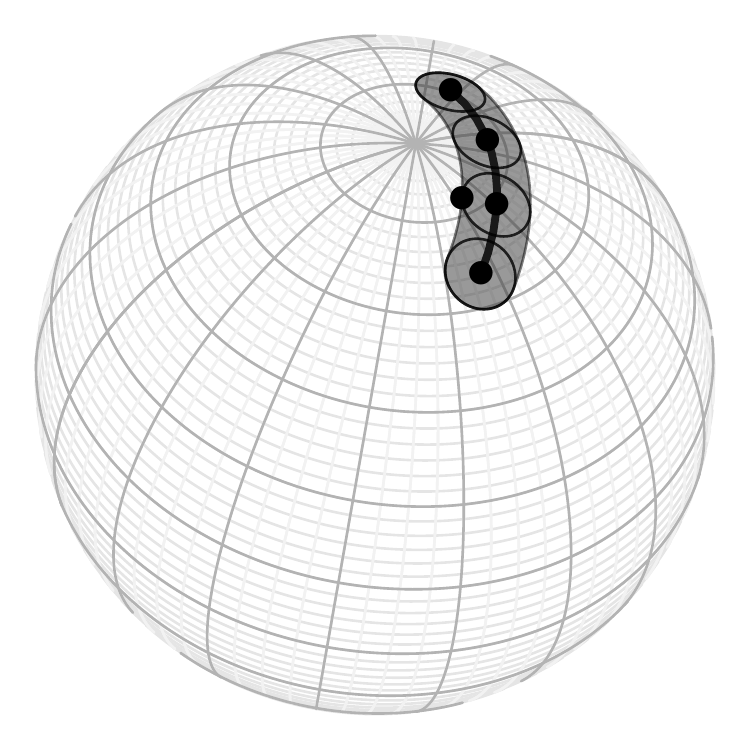}};
   \node [inner sep = 1pt] (tx1) at (40:1.5) {$\tilde{x}_{0.1}$};
   \draw [->,very thin] (tx1.west) to  [out=180,in=0] (0.26,0.84);
   \node [inner sep = 1pt] (tx2) at (25:1.5) {$\tilde{x}_{0.7}$};
   \draw [->,very thin] (tx2.west) to  [out=180,in=0] (0.366,0.69);
   \node [inner sep = 1pt] (tx3) at (10:1.5) {$\tilde{x}_{1.6}$};
   \draw [->,very thin] (tx3.west) to  [out=180,in=0] (0.395,0.5);
   \node [inner sep = 1pt] (tx4) at (-5:1.5) {$\tilde{x}_{10}$};
   \draw [->,very thin] (tx4.west) to  [out=180,in=0] (0.35,0.3);
   \node [inner sep = 1pt] (ty) at (175:1.5) {$\tilde{y}$};
   \draw [->,very thin] (ty.east) to  [out=0,in=180] (0.219,0.52);
 \end{tikzpicture}
 \hfill
 \includegraphics[scale=0.58]{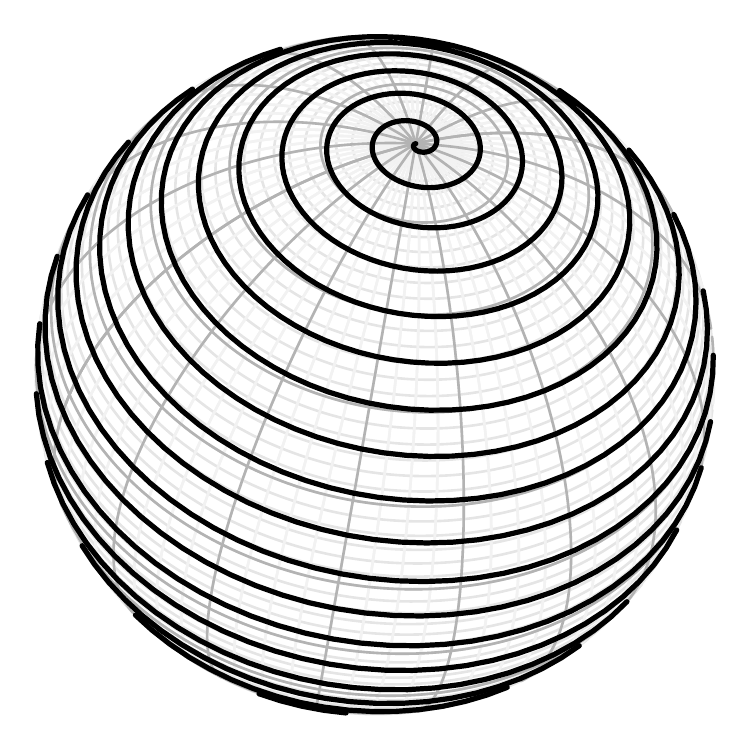}
 \hspace*{\fill}
 \caption{Left plot: Observed point $y$ (open circles) and the exponential curves $z \mapsto \exp(z / \gamma)$,
   with $z \in [0, 3]$ for the four values $\gamma = 0.1, 0.7, 1.6,
   10$ plotted in the zero-one standardization;
   the curve for $\gamma = 0.1$ is the lowest (most convex) curve,
   while the curve for $\gamma = 10$ is the topmost (almost linear) curve.
   Middle plot: tubular neighborhood $\mathbb{T}_r$ (gray area)
   around the curve $\mathbb{M} = \{\tilde{x}_\gamma : \gamma \in
   [0.1, 10]\}$ (black line) for the model
   $x_\gamma^{(k)} = \exp(z^{(k)} / \gamma)$
   with $z^{(k)} = k - 1$ for $k = 1, 2, 3$.
   Right plot: model curve for the
   counterexample at the end of Section~\ref{sec:chi2}.}
  \label{fig:tube}
\end{figure}

We now consider again the case of a single hypothesis $I = \{i\}$, but
generalize to the situation, where the parameter space $\Gamma_i$ no
longer consists of a single point. We
write $\Gamma = \Gamma_i$ and $x_{\gamma} = x_{\gamma, i}$.

By the considerations in the last section,
the LR statistic for a fixed value
$\gamma\in\Gamma$ equals $R_\gamma =
\tilde{x}_{\gamma}^{\top}\tilde{y}$, with
$\tilde{x}_{\gamma} = (B x_{\gamma}) /\norm{B x_{\gamma}}$.
Therefore, the LR statistic
of $H_0$ against $H_i$ equals $\inf_{\gamma\in\Gamma}{S(R_{\gamma})}
= S(R)$ with $R = \sup_{\gamma\in\Gamma} R_{\gamma}$
since $S$ is continuous and non-increasing in $R$.

% Appendix~\ref{sec:test-stat-single} shows
% that the likelihood-ratio statistic for $H_0$ against $H_{1}$
% is a monotonous function of the statistic $R = \tilde{x}^{\top}\tilde{y}$.

% Let $d = N -1$. Then under $H_0$ the statistic
% $\tilde{y}$ is uniformly distributed
% on the unit sphere $\mathbb{S}^d$ in the
% $d$ dimensional subspace
% $\mathbb{L}^d = \{a \in \mathbb{R}^{d+1}:
% a^{\top}\mathbf{1}_{d + 1} = 0\}$.
% In more detail, let the rows of the $d \times (d + 1)$ matrix
% $B$ form an orthonormal basis for $\mathbb{L}^d \subset \mathbb{R}^{d+1}$.
% Then $\tilde{x}^{\top} \tilde{y} = (B \tilde{x})^{\top}(B
% \tilde{y})$ and since
% $B C y$ is a vector of iid normal random variables,
% it follows that $B \tilde{y}$ is uniformly distributed on
% $\mathbb{S}^d$ under $H_0$.

Write again $r$ for the observed value of $R$,
define the model set
$\mathbb{M} = \{\tilde{x}_{\gamma} : \gamma \in \Gamma\}$
of the standardized prediction that are possible under $H_i$, and
define the tubular neighborhood
$\mathbb{T}_{r} = \mathbb{T}_{r}(\mathbb{M}) = \bigcup_{s \in \mathbb{M}}
\mathbb{C}_{s r}$ around $\mathbb{M}$.
Then it follows from the uniform distribution of
$\tilde{y}$ that
\[p = P_0(R > r) = P_0(\tilde{y} \in \mathbb{T}_{r}) =
\abs{\mathbb{T}_{r}} / \abs{\mathbb{S}}.\] The volume
$\abs{\mathbb{S}}$ of the unit sphere is straightforward to
calculate. \citet{hote:1939} gives explicit equations for
$\abs{\mathbb{T}_{r}}$ when $\mathbb{M}$ is a closed curve that
satisfies certain regularity conditions (essentially the nonlinear
transformation $x_\gamma$ needs to be sufficiently smooth) . These
results have been extended to more general manifolds $\mathbb{M}$;
see, for example, \citet{naiman:1990} and \citet{gray:2004}, and the
references therein.

The middle plot of Figure~\ref{fig:tube} illustrate this construction
for the exponential model $x_\gamma^{(k)} = \exp(z^{(k)} / \gamma)$
with $z^{(k)} = k - 1$ for $k = 1, 2, 3$.  Here and in the following,
we will write $x^{(k)}$ and $z^{(k)}$ for the $k$-th elements of the
vectors $x$ and $z$.  Assume that we observe a point $y = (-0.6,\,
-0.2,\, 0,\, 0.8)^{\top}$.  The correlation $R_{\gamma} =
\tilde{x}_{\gamma}^{\top} \tilde{y}$ is maximized when $\gamma = 1.7$.
Four points $\tilde{x}_{0.1},\, \tilde{x}_{0.7},\,
\tilde{x}_{1.6}\,\tilde{x}_{10}$, which are equally spaced on
$\mathbb{M}$, are shown together with their spherical caps; the
tubular neighborhood $\mathbb{T}_r$ is generated by moving the
spherical cap along the curve $\mathbb{M}$. Note that equal spacing of
points on $\mathbb{M}$ leads to unequally parameter values $\gamma$,
this is due to the nonlinearity of the model function.

Since $R_{\gamma}$ is invariant under affine transformations of
$x_{\gamma}$ and $y$, we may scale them to the unit interval
(``zero-one standardization''), as shown in the left-hand side of
Figure~\ref{fig:tube}.

\subsubsection{$\chi^2$ critical values }
\label{sec:chi2}

Before we continue to the case of multiple models in the next section, we
will consider what could happen if we would simply ignore the
identifiability issue and would assume that the statistic $-2 \log S$
is asymptotically $\chi^2$ distributed.

The critical value of such a test would depend only on the number of
parameters in the model, but not on the area of on sphere covered by
the model.  One can construct an example that shows that such a test
can not control the type-I error in general: in fact, for any critical
value $q \in \mathbb{R}_+$ of the statistic $-2 \log S$, there exists
a model so that the resulting test always rejects (has type-I error
1).  The idea is to choose the model in such way that the standardized
predictions form a curve that comes arbitrarily close to any point on
the unit sphere. In Appendix B we give the model
equation for a model function that fulfills these requirements, see
also the right-hand side of Figure~\ref{fig:tube}, which gives a
graphical illustration of this model.

\subsubsection{Multiple nonlinear dose-response models}
\label{sec:multmod}

The case of multiple models can be treated along the lines of a single
model, just by forming the union of the dose-response model shapes on
the unit sphere: the multiple models can be collapsed into a single
big model, so that the case of a candidate set of models is along the
lines of the single model.  From a multiple testing perspective a
max-test is being performed, \textit{i.e.} the maximum being over the
different candidate dose-response models.

In more detail, the LR statistic is a monotonous function of the
statistic
\[R = \max_{i \in I} \sup_{\gamma_i \in \Gamma_i}
    \tilde{x}_{\gamma_i, i}^{\top} \tilde{y},\qquad\text{with}\
    \tilde{x}_{\gamma_i, i} = (B x_{\gamma_i, i}) /
    \norm{B x_{\gamma_i, i}}.\]
Equivalently, we can write $R$ as $R = \sup_{\gamma \in \Gamma'} \tilde{x}_{\gamma}^{\top} \tilde{y}$
for a single (composite) model
\[
\Gamma' = \bigcup_{i \in I} \bigl\{\Gamma_i \times \{i\}\bigr\}, \qquad
\tilde{x}_{(\gamma_i, i)} = \tilde{x}_{\gamma_i, i}.
\]
If we define
$\mathbb{M} = \{\tilde{x}_{\gamma} : \gamma \in \Gamma'\}$
and $\mathbb{T}_{r}(\mathbb{M})$
for this composite model, then the task reduces to the calculation
of the volume $\abs{\mathbb{T}_{r}}$, since
\[
P_0(R > r) = \abs{\mathbb{T}_{r}} / \abs{\mathbb{S}}.
\]

Analytic methods developed for the case of a single model can no
longer be used as the different models might have intersecting curves
on the unit-sphere (see Figure~\ref{fig:mult} for an illustration),
which is why specialized numerical methods need to be developed to
calculate the associated tubular volumes. 

In addition to the overall p-value $p = P_0(R > r)$, the
multiplicity-adjusted p-values $p_i = P_0(R > r_i)$ for the individual
hypothesis $H_i$, for $i \in I$, are also of interest, where $r_i$ is
the correlation between the observations and the best predictions
under the hypothesis $H_i$.

This ``multiplicity'' penalty has a direct geometric interpretation:
if the candidate set of dose-response models gets larger, the
standardized predictions cover larger parts on the unit sphere;
consequently, the tube around the standardized model predictions has
to get thinner if the volume $\abs{\mathbb{T}_{r}}$ (and therefore the
type-I error) should be kept constant.

\begin{figure}
  \centering
  \begin{tikzpicture}[scale = 2,>=stealth]
     \pgftext {\includegraphics[scale=0.29]{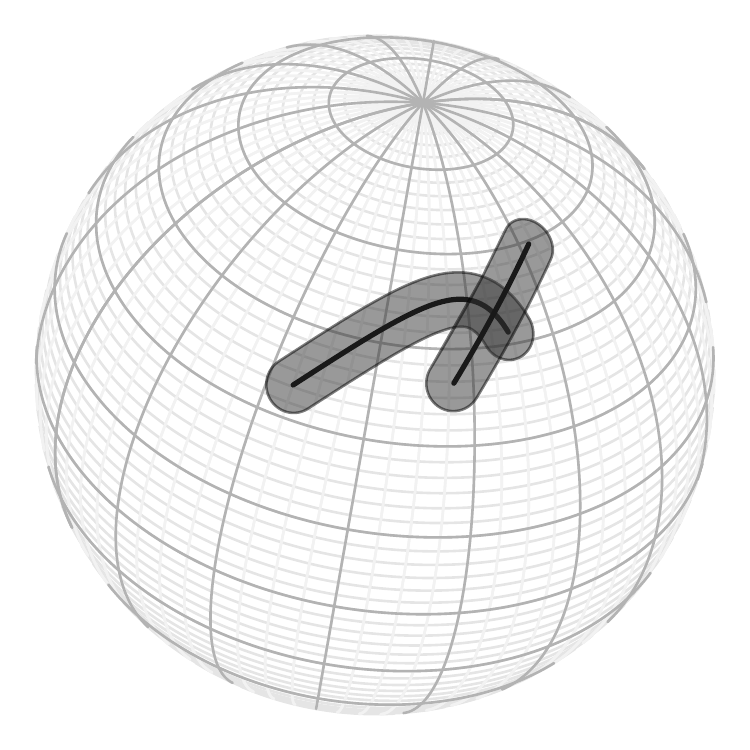}};
     \node [inner sep = 1pt] (M1) at (30:1.5) {$\mathbb{M}_{1}$};
     \draw [->,very thin] (M1.west) to  [out=180,in=45] (0.457,0.395);
     \node [inner sep = 1pt] (T1) at (15:1.5) {$\mathbb{T}_{r1}$};
     \draw [->,very thin] (T1.west) to  [out=180,in=0] (0.525,0.35);
     \node [inner sep = 1pt] (M2) at (160:1.5) {$\mathbb{M}_{2}$};
     \draw [->,very thin] (M2.east) to  [out=0,in=140] (0,0.14);
     \node [inner sep = 1pt] (T2) at (175:1.5) {$\mathbb{T}_{r2}$};
     \draw [->,very thin] (T2.east) to  [out=0,in=180] (-0.33,-0.028);
  \end{tikzpicture}
 \caption{Standardized predictions
   $\mathbb{M}_i = \{\tilde{x}_{i,\gamma_i} : \gamma_i \in \Gamma_i\}$
   and tubular neighborhoods $\mathbb{T}_{r i} = \mathbb{T}_r(\mathbb{M}_i)$
   for the models
   $x_{1,\gamma}^{(k)} = \cos(z^{(k)} + \gamma)$
   and
   $x_{2,\gamma}^{(k)} = (z^{(k)})^{\gamma} /
   (z^{(k)})^{\gamma} + 1.5)$, where
   $\Gamma_1 = [\pi, (5/4)\pi]$ and $\Gamma_2 = [10^{-3},5]$,
   when $z = (0, \dots, 3)^{\top}$.
   Due to the overlap between the two neighborhoods,
   the volume of the combined tube, $\abs{\mathbb{T}_{r1} \cup
   \mathbb{T}_{r2}}$, is smaller then the sum of the
   two volumes, $\abs{\mathbb{T}_{r1}} + \abs{\mathbb{T}_{r2}}$.}
  \label{fig:mult}
\end{figure}
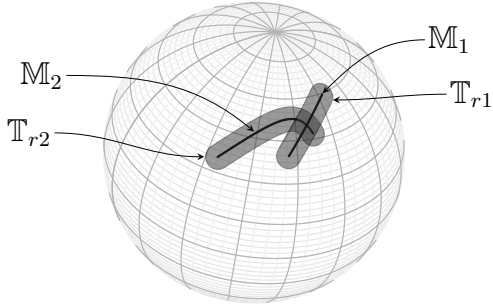

\subsection{Numerical calculation of tubular volumes}
\label{sec:num}

This section describes an algorithm motivated by importance sampling
to approximate the probabilities $P_0(R > r) = P_0(\tilde{y} \in
\mathbb{T}_r)$, to obtain p-values or a critical value $r_{crit}$ for
the test statistic $R$. As discussed earlier, under the null
hypothesis, $\tilde{y}$ is uniformly distributed on the sphere, so
that $P_0(R > r)=\int_{\mathbb{T}_r} \mathrm{d}\varsigma$, with
$\varsigma$ the uniform probability measure on $\mathbb{S}$. In
addition to p-values and the critical value, we are also be interested
in calculating the power $P_1(R > r) = P_1(\tilde{y} \in
\mathbb{T}_r)$, for a critical value~$r$, with $P_1(\cdot)$ a
probability calculated under some alternative hypothesis.  By
arguments due to \citet{pukkila:1988}, under alternative hypotheses,
$\tilde{y}$ follows an angular Gaussian distribution on $\mathbb{S}$,
which means that $P_1(R > r)=\int_{\mathbb{T}_r} f\,\mathrm{d}
\varsigma$, where $f : \mathbb{S} \to \mathbb{R}$ is the density of
the angular Gaussian density (of course, under the null-hypothesis $f
\equiv 1$.)

Let us motivate the algorithm by considering a naive sampling
approach.  If we draw a sample $U_1, \dots, U_{\kappa}$
%from~$\varsigma$,
uniformly from $\mathbb{S}$, then $\int_{\mathbb{T}_r} f\,\mathrm{d}
\varsigma \approx \sum_{k = 1}^{\kappa} \ind{U_k \in \mathbb{T}_r}
f(U_k) / \kappa$.  However, since $\mathbb{T}_r$ is only implicitly
defined by the nonlinear models, deciding if a point $U_k$ belongs to
$\mathbb{T}_r$ requires the computationally expensive calculation of
the maximum-likelihood estimates for each of the nonlinear models.

A better approach is to restrict sampling to $\mathbb{T}_r$
instead of sampling from the whole sphere $\mathbb{S}$.
If one can construct a probability distribution
supported only on $\mathbb{T}_r$ with known density $g$
with respect to $\varsigma$,
and generate
a sample $V_1, \dots, V_{\kappa}$
according to this distribution, then we could use
importance sampling to obtain the approximation
\begin{equation}
  \label{eq:impsamp}
  \int_{\mathbb{T}_r} f\, \mathrm{d} \varsigma
  \approx \frac{1}{\kappa} \sum_{k = 1}^{\kappa} \frac{f(V_k)}{g(V_k)}.
\end{equation}
However, calculating such density $g$ is challenging.
The approach proposed in this paper is to sample points $V_1, \dots,
V_{\kappa}$ in such a way that even if we can not calculate $g$
exactly, we can approximate $g$ from the sample itself.  The main idea
for sampling only within $\mathbb{T}_r$ is to first sample a model $i$
from the $m$ models with equal probability, then sample a value
$\gamma$ from $\Gamma_i$ according to some convenient distribution
(for example, the uniform distribution if $\Gamma_i$ has finite
volume) and take $W = \tilde{x}_{\gamma, i}$. Then one samples a point
$V$ from $G_W$, the uniform probability distribution on the spherical
cap~$\mathbb{C}_{W r}$ around $W$ with radius $r$.  It is possible to
approximate the density $g$ underlying this proposal sampling
mechanism (see the Appendix C for details), which is needed to
calculate the importance ratios in (\ref{eq:impsamp}).

Note that in the strict sense however this is not an importance
sampling algorithm, as the proposal density $g$ as well as integral
are determined from the same sampled values. In Appendix C the
consistency of the sampling scheme is proved. That means that the
approximation error of the algorithm can be made arbitrarily small by
increasing the number of sampling replicates $\kappa$. The number of
simulations can be chosen by monitoring the Monte Carlo standard
error, to obtain the desired precision.

\subsection{Comparison to MCP-Mod}
\label{sec:comparison-mcp-mod}

The multiple contrast test of the MCP-Mod procedure
~\citep{bret:pinh:bran:2005} can be viewed as a discrete analogue of
the likelihood-ratio approach presented here, as
each model is restricted to a finite number of possible
shapes in MCP-Mod.

Even in the case of a finite number of possible shapes (when both
approaches are applicable), the likelihood-ratio test seems preferable
due to the following argument.  The t-test for testing $\beta = 0$ is
uniformly most powerful invariant with respect to affine
transformations \cite[Chapter 7]{lehm:roma:2008}.  While the
likelihood-ratio test is equivalent to the t-test, the corresponding
contrast test statistic is invariant but differs from the t-test in
the variance estimate; hence it is not in general admissible (although
in this setting power gains by an LR test are unlikely to be large).

\section{Applications}
\label{sec:appl}

In this section we will use data from a dose-response clinical trial
to illustrate the methodology.  One of the objectives of such trials
is to test for a dose-response effect, that is, whether the
dose-response is flat or not.  Motivated by the data, we define some
scenarios and compare the approach to alternative trend tests: Section
\ref{sec:sim} considers a single candidate model; Section
\ref{sec:power-calc-mult} considers a candidate set of models.

In all examples considered p-values and power values were calculated
using a Monte Carlo standard error of at most $0.001$, or when a maximum
sample of $10^6$ samples was reached in the algorithm. Calculation of the
critical value was done using root-findung under the null-hypothesis.

\subsection{Dose-Finding Example}
\label{sec:data}

The purpose of this section is to illustrate the methodology on a real
data set and compare it to the MCP-Mod methodology. The section also
illustrates how the critical value (and thus the multiplicity
adjustment) depends on the complexity of the candidate set of models.

The used data set is available in the \texttt{DoseFinding} R package
under the name \texttt{biom}.  The data comes from 100 patients
allocated equally to a placebo and four treatment arms with
dose levels 0.05,\, 0.2,\, 0.6,\, 1.

To observe the impact of the richness of the candidate set on the
critical value (\textit{i.e.} the resulting multiplicity adjustment),
we will sequentially increase the set of candidate models. We start
with the Emax model $\alpha + \beta \,z^{(k)} / (z^{(k)} + \gamma)$,
where the interval $\Gamma = [0.001, 1.5]$ was chosen for the
parameter $\gamma$. This interval covers a wide range of possible
shapes underlying the Emax model. This can be seen, when plotting and
overlaying the resulting (``zero-one'' standardized) response curves;
see Figure~\ref{fig:example}\,(a). The boundaries of the polygon
correspond to $\gamma$ equal to $0.001$ and $1.5$. When using a
significance level of 5\% one-sided, the resulting critical value is
$0.197$.  When raising the upper limit of $\gamma$ to $10$, although
this results in a much larger interval for $\gamma$, the critical
value only goes up to $0.199$. The reason is that the standardized
model predictions do not cover much additional area on the unit
sphere.  This can also be seen in Figure~\ref{fig:example}\,(a) in the
dark gray area: Due to the nonlinearity, the additional flexibility on
the parameter space (by increasing the upper bound to 10) does not
lead to a major increase in flexibility of the model shapes.

\begin{figure}[t!]
  \begin{center}
    \includegraphics[width=0.95\textwidth]{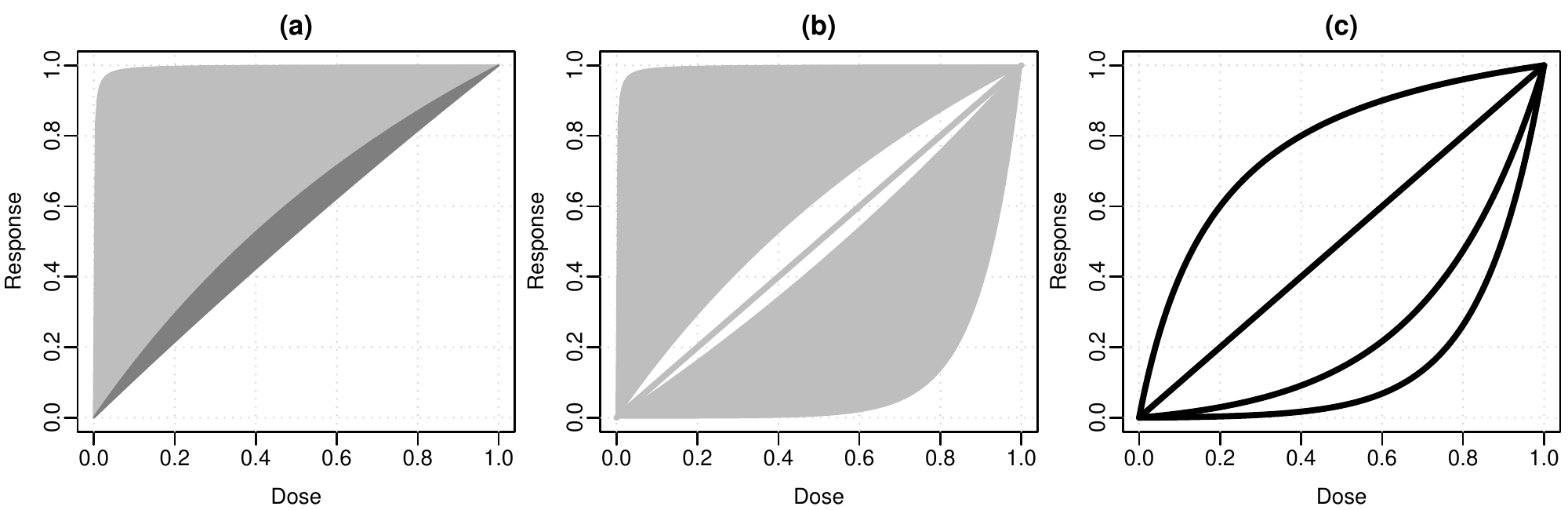}
  \end{center}
  \caption{Response curves in the candidate set. The curves were
    standardized to have effect 0 at dose-level 0 and effect 1 at
    dose-level 1 (the ``zero-one'' standardization from
    Figure~\ref{fig:tube}): (a) depicts Emax models for $[0.001, 1.5]$
    (gray area) and $[0.001, 10]$ (gray and dark gray area
    combined). (b) depicts curves corresponding to the used candidate
    models: Emax, linear and exponential (c) depicts the candidate
    curves used for the MCP-Mod methodology}
  \label{fig:example}
\end{figure}

Now suppose we would like to add a linear model $\alpha + \beta
z^{(k)}$ to the candidate set of models. The linear model does not use
a dose transformation and it represents only a single shape; see also
the linear increasing line in Figure~\ref{fig:example}\,(b).
Therefore, adding this model does not make the candidate set much
broader in terms of shapes, and the critical value only increases to
$0.200$.  When adding an exponential model of form
$\exp(z^{(k)}/\gamma)-1$ with parameter bounds $[0.1,2]$, this leads
to a more pronounced extension of the possible predictions.  This is
also reflected in the critical value, which goes up to $0.210$.
Figure~\ref{fig:example}\,(b) shows the possible shapes that will be
used for the LR test.

So there is a direct, intuitive connection between the critical value
and the complexity of the candidate model shapes. The more parts of
the unit sphere are covered the larger the multiplicity penalty.  On
the other hand, when an additional shape is added that is similar to
other shapes that are already in the candidate set, the critical value
does not increase noticeably, due to the overlap of tubes.

Fitting the dose-response models with the necessary parameter
constraints can be done using the \texttt{fitMod} function from the
\texttt{DoseFinding} R package. The fitted dose-response functions are
shown in Figure \ref{fig:figbiom}

\begin{figure}[t!]
  \begin{center}
    \includegraphics[width=0.95\textwidth]{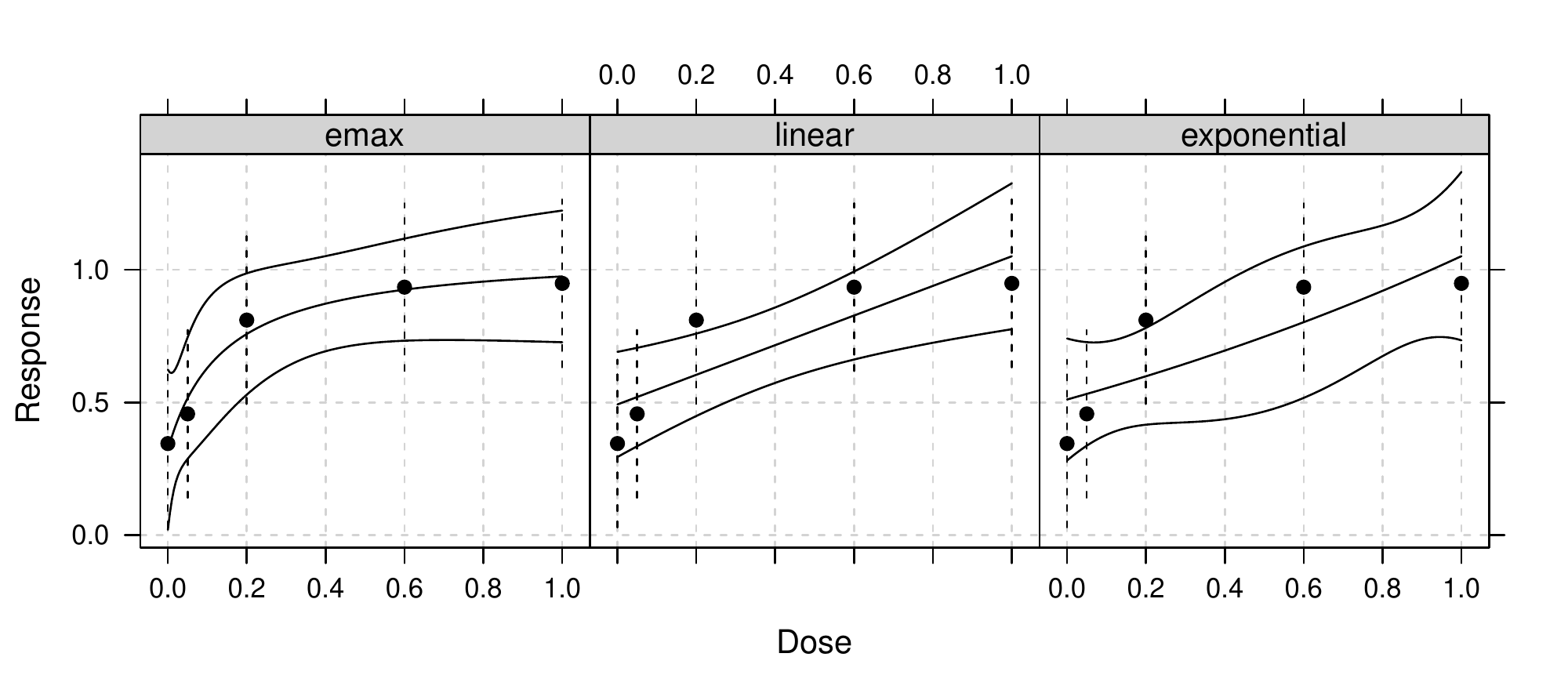}
  \end{center}
  \caption{Fitted dose-response functions for biom data-set with 95\%
    pointwise confidence intervals, the means at each dose with 95\%
    confidence intervals are also shown.}
  \label{fig:figbiom}
\end{figure}

The parameter estimates of the models can be found in
Table~\ref{tab:resultslr}, in addition to the p-values for testing
$\beta=0$ versus $\beta>0$, for each of the three models.  It can be
seen that the p-values for the models, considering all three models as
candidate set, are all smaller than $0.01$.  Of course, this is also
true for the p-values that result from considering each model class
separately.

\begin{table}
\centering
\caption{Results of likelihood-ratio test.
  The test statistic $R$ is the
  correlation between the observation and the best model prediction.
  p-value$^*$ is the p-value taking into account the fact that a
  candidate set of models was used (multiplicity adjustment for
  the multiple models), while p-value$^{**}$ is the p-value
  within the respective model class (no adjustment for the multiple models).}
\label{tab:resultslr}
\begin{tabular}{@{}llrrr@{}}
  \toprule
Model & Parameter estimates & test-statistic & p-value$^*$ & p-value$^{**}$ \\
  \midrule
Emax ($\gamma \in [0.001,1.5]$) & $\alpha=0.32,\, \beta=0.75,\, \gamma=0.14$ & 0.335 & 0.001 & 0.001 \\
  Linear & $\alpha=0.49,\, \beta=0.56$ & 0.287 & 0.006 & 0.002 \\
  Exponential ($\gamma \in [0.1,2]$) & $\alpha=0.51,\, \beta=0.83,\, \gamma=2.00$ & 0.276 & 0.009 & 0.004 \\
   \bottomrule
\end{tabular}
\end{table}

For comparison we will also apply the MCP-Mod procedure.  Similar to
\cite{bret:pinh:bran:2005}, we choose an Emax shape with $\gamma=0.2$,
a linear shape and two exponential shapes, one with
$\gamma=0.5/\log(6)$ and the other with $\gamma=0.15$. The candidate
shapes of the MCP-Mod procedure are presented in
Figure~\ref{fig:example}\,(c). The results shown in
Table~\ref{tab:resultsmm} have been calculated using the
\texttt{MCTtest} function in the \texttt{DoseFinding} package. One can
see that the Emax and linear model have p-values $<0.01$ similar to
the LR test. However, the two exponential shapes have p-values
$>0.025$, despite the fact that a trend could be detected using the LR
test for the exponential model. The reason is that neither of the two
$\gamma$ values fits the data well. If the exponential model for
$\gamma=2$ would have been included in the set of MCP-Mod candidates
also a p-value smaller than $0.01$ would be observed.

After establishing the existence of a dose-response effect, one can
continue by either selecting or averaging models to estimate the
dose-response curve and the target dose of interest; a detailed
discussion of these topics can be found in \cite{scho:born:bret:2015}.

\begin{table}
\centering
\caption{Results of MCP-Mod testing approach.}
\label{tab:resultsmm}
\begin{tabular}{@{}lrr@{}}
  \toprule
Shape & test-statistic & p-value \\
  \midrule
Emax ($\gamma=0.2$) & 3.464 & 0.001 \\
  Linear & 2.972 & 0.004 \\
  Exponential ($\gamma=0.15$) & 2.218 & 0.028 \\
  Exponential ($\gamma=0.5/ \log(6)$) & 1.898 & 0.056 \\
   \bottomrule
\end{tabular}
\end{table}

% \textcolor{red}{TODO: Sollen wir den folgenden Text in die Einleitung
%   verschieben?}

% \textcolor{green}{Bis jetzt erwähnen wir MCP-Mod nicht in der
%   Einleitung, da wir das Paper ja so allgemein wie möglich halten
%   wollten. D.h. ich würde dazu tendieren den Text unten zu löschen,
%   es sei denn wir wollen das Paper mehr in Richtung dose-response und
%   MCP-Mod ausrichten.
% }

% \textcolor{red}{
% The MCP-part of MCP-Mod can be seen as a special case of the general
% approach presented here: A candidate set $\Gamma'$ of models is used,
% but the sets $\Gamma_i$ in the candidate sets are only allowed to
% exist out of single points (singletons). In this situation one can use
% distributional results based on the multivariate t-distribution to
% calculate the critical values and p-values for the maxmimum
% statistic. Note that the methodology presented in this paper allows
% intervals or more generally compact sets to be used for the
% $\Gamma_i$, which might be easier to pre-specify a-priori than fixed
% single points. The methodology proposed in this paper can hence be
% seen as a continuous extension of MCP-Mod, where
% the number of candidates is allowed to go to infinity.
% }

% \textcolor{red}{
% In addition in the Mod part of MCP-Mod, where the dose-response models
% are fitted to the data for dose and dose-response estimation,
% continuous intervals or more generally compact sets are used for
% $\Gamma_i$, which introduces an inconsistency between the models used
% in MCP and Mod part of MCP-Mod.
% }

\subsection{Power calculations for a single Emax model}
\label{sec:sim}

In this section, we compare the power of the LR test to tests that are
optimal for specific values of $\gamma$: assume that the true model is
an Emax model with a parameter value $\gamma$; then, as discussed in
Section \ref{sec:comparison-mcp-mod}, the t-test of $\beta = 0$ versus
$\beta > 0$ in the linear model $y \sim \mathcal{N}(\alpha \bone{n} +
\beta x_{\gamma}, \sigma^2 I_n)$ using the true parameter $\gamma$, is
uniformly most powerful invariant. We will refer to this test as
locally optimal for $\gamma$.  These tests provide a useful upper
bound for the performance of the likelihood-ratio test.

Consider again the dose-response trial with dose levels $0,\, 0.05,\,
0.2,\, 0.6,\, 1$.  We take 20~observations per dose level and choose
the non-centrality parameter $\delta = \beta \norm{B x_{\gamma}} /
\sigma$ \citep[Section 6]{seber:2003} so that the locally optimal test
has power $80\%$ for a one-sided type I error of $5\%$.

For scenarios with different $\gamma$ values from $\Gamma = [0.001,
1.5]$, we look at the power of a LR test with $\gamma \in \Gamma$, and
the power of locally optimal tests for the four parameter values
$\gamma = 0.001,\, 0.035,\, 0.159,\, 1.5$. These parameter values are
chosen so that the corresponding points $x_{\gamma}$ are separated by
equal distances along the model curve $\mathbb{M} =
\{\tilde{x}_{\gamma} : \gamma \in \Gamma\}$. The critical value of the
LR test is given by $0.197$. In Figure~\ref{fig:emaxpower} it can be
seen that the LR test achieves a power above 70\% over the whole range
of $\Gamma$ and is rather close to the respective locally optimal
test. The power of the each locally optimal tests decreases markedly,
when the parameter $\gamma$ is mis-specified. For example the optimal
test for $\gamma=0.001$ has only around 50\% power when the true value
is $\gamma=0.24$.

\begin{figure}
  \centering
  \includegraphics[scale=0.7]{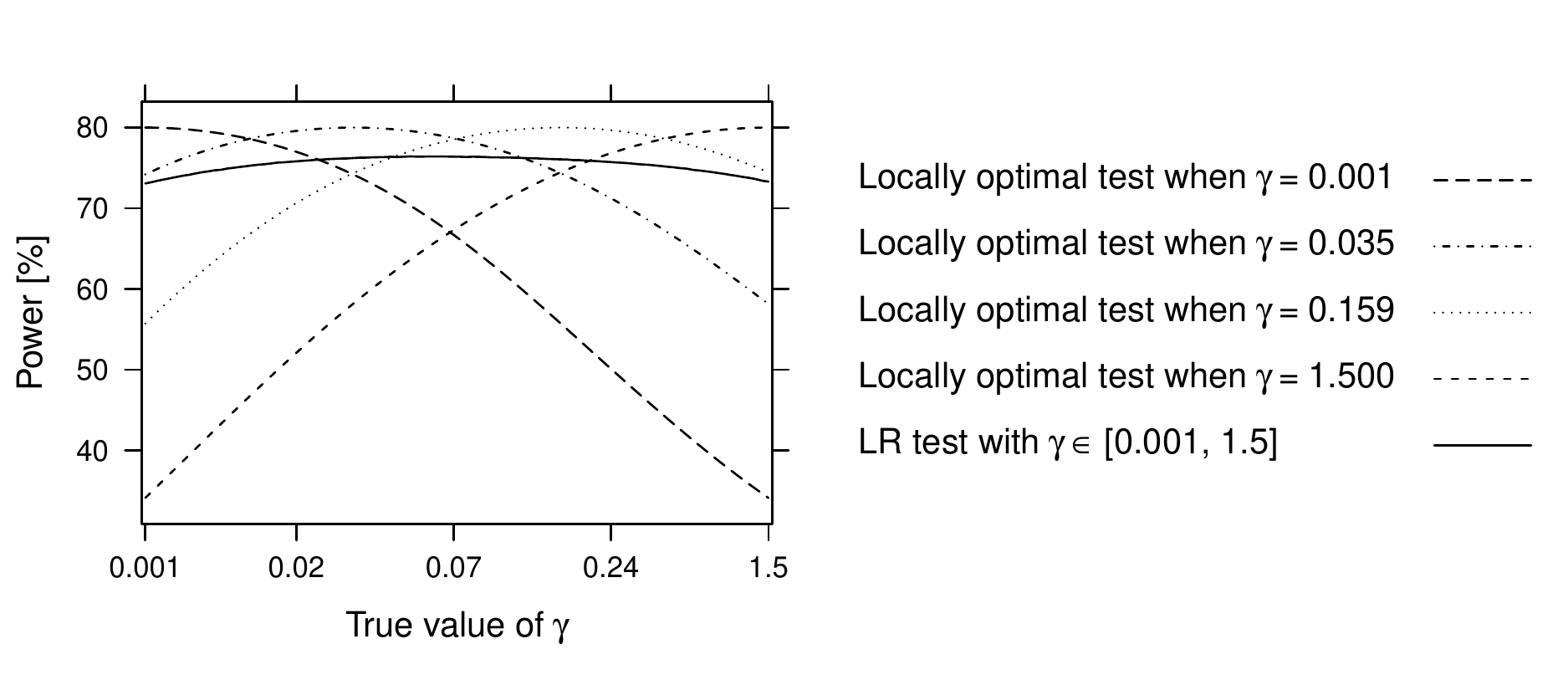}%
%  \hfill\includegraphics[viewport=3 11 302 311, scale=0.75]{fig80}
  \caption{
Power values of different test procedures
in dependence of the true parameter value $\gamma$
in the Emax model $x_{\gamma}^{(k)} = z^{(k)} / (z^{(k)} + \gamma)$
with 20 observations for each of the dose levels
$0,\, 0.05,\, 0.2,\, 0.6,\, 1$.
Power values of the locally optimal tests for
$\gamma$ values $0.001, \, 0.035,\, 0.159,\, 1.5$
are shown with dotted and dashed lines.
Power values of the likelihood-ratio test
with $\Gamma = [0.001, 1.5]$ are shown as solid line.
The axis for the true $\gamma$ values
is scaled so that equal distances on the axis
correspond to equal distances along the model curve
$\gamma \mapsto \tilde{x}_\gamma$ of standardized predictions.}
  \label{fig:emaxpower}
\end{figure}

\subsection{Power calculations for multiple models}
\label{sec:power-calc-mult}

We now define five scenarios for the true underlying mean vector: 1) a
linear model; 2) an Emax model with parameter $\gamma = 0.2$; 3) an
exponential model with parameter $\gamma = 0.1$; 4) an exponential
model with parameter $\gamma = 0.5 / \log(6)$; and
5) a sigmoid Emax model of the form $x = z^ 4/(z^4+0.05^4)$.

We choose the non-centrality parameter so that in each scenario, the
locally optimal test has a power of $50\%$ or $80\%$ for a one-sided
type I error of $5\%$. Note that these locally optimal tests are in
practice ``unachievable'', as they use information about the true
dose-response model class and the true parameter $\gamma$.

The LR test uses the following candidate models: the linear model, the
Emax model with $\Gamma = [0.001, 1.5]$, and the exponential model
with $\Gamma = [0.1, 2]$. The critical value of the LR test is given
by $0.210$.

For comparison, we consider three multiple contrast tests. First we
use Williams contrasts and Marcus contrasts \citep{will:1971,
  marc:1976}, as implemented in the \texttt{multcomp} R package. Both
are known to be powerful trend tests. In addition, we use a multiple
contrast test with four model-based contrasts (as in MCP-Mod), where
the contrasts are optimized to detect the true underlying simulation
scenarios 1-4 (see \cite{pinh:born:glim:2014} for details on how to
calculate these optimal contrasts). This is unrealistic, as in
practice the value of $\gamma$ for each model is unknown, but it gives
a useful benchmark.  Scenario 5) is included to investigate the
behaviour for a model shape that is neither part of the LR test nor
MCP-Mod set of candidate shapes.

\begin{table}
\caption{Power (in percent) of different test procedures
The five scenarios are: 1) a linear
model; 2) an Emax model with parameter $\gamma = 0.2$; 3) an
exponential model with parameter $\gamma = 0.1$; 4) an exponential
model with parameter $\gamma = 0.5 / \log(6)$; and 5) a
sigmoid Emax model of the form $x = z^ 4/(z^4+0.05^4)$.
The  non-centrality parameter is chosen so that a local optimal test
has either 50\% power (first 5 rows) or 80\% power (last 5 rows).}
\label{tab:resultsmod}
\centering\setlength{\tabcolsep}{8pt}
\begin{tabular}{@{}cl*{8}{c}@{}}
  \toprule
  && & \multicolumn{4}{c}{Locally optimal test for scenario}\\
  \cmidrule(l){4-7}
  Power & Scenario & LR  & 1   & 2   & 3   & 4   & MCP-Mod  & Williams & Marcus\\
  \midrule
50 & 1 (Linear) &43.3&50.0&44.1&39.4&45.1&46.8&34.9&43.0\\
   & 2 (Emax)   &43.4&44.1&50.0&25.3&32.1&44.1&41.1&43.8\\
   & 3 (Exp 1)  &39.4&39.4&25.3&50.0&48.8&44.1&28.0&37.7\\
   & 4 (Exp 2)  &41.6&45.1&32.1&48.8&50.0&46.1&30.6&40.2\\
   & 5 (Sigm)   &41.2&30.7&44.8&15.1&19.5&36.2&44.0&41.2\\\addlinespace
80 & 1 (Linear) &73.4&80.0&73.0&66.6&74.3&76.8&61.5&72.9\\
   & 2 (Emax)   &73.4&73.0&80.0&43.0&55.1&74.5&69.7&73.5\\
   & 3 (Exp 1)  &69.9&66.6&43.0&80.0&78.6&74.5&52.2&67.8\\
   & 4 (Exp 2)  &72.1&74.3&55.1&78.6&80.0&76.2&56.0&70.3\\
   & 5 (Sigm)   &71.1&52.9&73.9&23.2&31.9&65.0&73.6&71.0\\
\bottomrule
\end{tabular}
\end{table}

In Table \ref{tab:resultsmod} one can observe that the performance of
the locally optimal tests decreases for the scenarios they are not
optimized for, only the locally optimal test corresponding to the
linear shape gives a surprisingly good overall performance. Among the
multiple contrasts tests the MCP-Mod contrasts, which use information
about the true shapes, perform best, apart from the mis-specified
scenario 5, where both Williams and Marcus contrasts perform better.

The LR test gives a more robust overall performance. For scenarios 1-4
the power is slightly lower than for the MCP-Mod contrasts, despite
the fact that a richer candidate set of models is used. For scenario 5
one can however see that the LR test outperforms the MCP-Mod
contrasts. This is most is due to increased robustness (\textit{i.e.}
increased flexibility in model shapes) of the LR test compared to the
MCP-Mod contrasts, so that for this shape not included in the
candidate set a better performance is obtained.

\section{Conclusions}
\label{sec:concl}

In this manuscript, we have considered the problem of detecting a
dose-related trend based on a set of candidate dose-response models.

The problem is not identifiable asymptotically and the standard
asymptotic $\chi^2$ distribution of the likelihood-ratio statistic
does not apply. Furthermore one can show that using a critical value
that depends only on the number of parameters in the candidate models
and not the model complexity (\textit{i.e.} the parts covered on the
unit sphere) may lead to an arbitrarily large type-I error inflation.
To avoid this, we work with the exact small-sample distribution of the
likelihood-ratio statistic.  Based on a geometric interpretation of
the test statistic due to Hotelling, an sampling algorithm has been
developed to approximate the exact distribution of the test statistic.

This work extends the previously available methods for dose-response
testing in several respects.  The multiple contrast tests in
MCP-Mod use a fixed set of guesstimates of the nonlinear parameters
$\gamma$ in the testing step, but then in the modelling step estimates
those parameters from the data. We do not require that the parameters
in the nonlinear model part are fixed guesstimates, but allow to vary
in a defined interval.

An advantage over alternative approaches to derive the distribution of
the likelihood ratio test statistic under multiple models (such as
those in \cite{dett:tito:bret:2015} and \cite{baay:houg:2015}) is that
we work with the finite sample distribution instead of the asymptotic
distribution. In addition to rejection probabilities under the null
hypothesis, we also consider rejection probabilities under
alternatives. This allows to perform power calculations and thus
sample-size calculations at the design stage of an experiment, which
is of crucial importance in clinical trials.

The developed methods for calculation of the distribution of the LR
test could be of interest beyond dose-response analysis. Nonlinear
models, where a test of trend is of interest, appear in many areas of
applied sciences such as biology and economics (\textit{e.g.} change
point analysis or harmonic regression).

We have assumed that the residuals are normally distributed.  The
extension to elliptically contoured distributions is possible
directly~\citep{fang:1990, gupt:1993}; for other distributions,
components of the likelihood-ratio statistic may not be restricted to
the unit sphere anymore, but the general approach may still be applied
in some situations: for example, \citet{diac:1985} show that the
rejection regions of a test for independence in a two-way contingency
table are tubular neighborhoods on a simplex on which the test
statistic is uniformly distributed under the null hypothesis.

We have also assumed that under the correct model, residuals are
independent. The extension to the case of a known correlation
structure appears straightforward. More fundamental extensions, such
as allowing random effects in addition to the fixed effects, could
also be considered.

Another extension would be to test more complicated null
hypotheses. For example, a test for a hypothesis $H_{I}$ versus an
alternative $H_{J}$ where $I \subseteq J \subseteq \{1, \dots, 2 m\}$
allows to test individual parameters of a nonlinear model, such as the
Hill parameter of the sigmoid Emax model.

The rationale for developing a specialized numerical algorithm for
this problem is that direct simulation under the null hypothesis is
computationally difficult, since it involves multiple iterative
optimizations for each simulation replicate with a poorly identified
nonlinear parameter $\gamma$.  By using a form of importance sampling
it is not necessary to perform nonlinear optimization for the sampling
replicates. Note that the developed sampling algorithm might be of
interest in general for calculation of the volume of tubes.  At the
moment dimensionality that the algorithm works on grows with the
number of observations, an improvement of the algorithm would be to
work with sufficient statistics instead of raw data. This makes the
algorithm more efficient, but also more complicated. We also note here
that it is numerically advantageous to replace random samples by
quasi-random samples (as discussed for example in
\cite{fang:wang:1994}), in which case quasi importance sampling gives
a quadrature rule for functions on tubular neighborhoods.

This paper is primarily concerned with testing for a dose-response
effect, which is only one of the question of interest in dose-finding
studies. Once a dose-response effect has been established the
following questions are to estimate the dose-response curve and target
doses of interest, so a further topic to explore is the relation
between testing and estimation. For (frequentist) model averaging, the
predictions are most commonly weighted according to either the AIC or
the BIC of the models (see among others
\cite{scho:born:bret:2015}). Both of these criteria penalize models
only according to the number of parameters in the model.  We have seen
that the number of parameters is, at least for testing, a poor
surrogate for the model complexity (the flexibility of the predictions
that are possible under a model).  Maybe geometric considerations can
be used to penalize complex models in a more meaningful way.

\bibliographystyle{abbrvnat}

\begin{thebibliography}{37}
\providecommand{\natexlab}[1]{#1}
\providecommand{\url}[1]{\texttt{#1}}
\expandafter\ifx\csname urlstyle\endcsname\relax
  \providecommand{\doi}[1]{doi: #1}\else
  \providecommand{\doi}{doi: \begingroup \urlstyle{rm}\Url}\fi

\bibitem[Andrews(1996)]{andrews:1996}
D.~Andrews.
\newblock Admissibility of the likelihood ratio test when the parameter space
  is restricted under the alternative.
\newblock \emph{Econometrica}, 64:\penalty0 705--718, 1996.

\bibitem[Andrews and Ploberger(1994)]{andr:plob:1994}
D.~Andrews and W.~Ploberger.
\newblock Optimal tests when a nuisance parameter is present only under the
  alternative.
\newblock \emph{Econometrica}, 62:\penalty0 1383--1414, 1994.

\bibitem[Andrews and Ploberger(1995)]{andrews:1995}
D.~Andrews and W.~Ploberger.
\newblock Admissibility of the likelihood ratio test when a nuisance parameter
  is present only under the alternative.
\newblock \emph{The Annals of Statistics}, 23:\penalty0 1609--1629, 1995.

\bibitem[Baayen et~al.(2015)Baayen, Hougaard, and Pipper]{baay:houg:2015}
C.~Baayen, P.~Hougaard, and C.~B. Pipper.
\newblock Testing effect of a drug using multiple nested models for the
  dose-response.
\newblock \emph{Biometrics}, 71:\penalty0 417--427, 2015.

\bibitem[Bornkamp et~al.(2009)Bornkamp, Pinheiro, and
  Bretz]{born:pinh:bret:2009}
B.~Bornkamp, J.~C. Pinheiro, and F.~Bretz.
\newblock {MCPMod}: An {R} package for the design and analysis of dose-finding
  studies.
\newblock \emph{Journal of Statistical Software}, 29\penalty0 (7):\penalty0
  1--23, 2009.

\bibitem[Bretz et~al.(2005)Bretz, Pinheiro, and Branson]{bret:pinh:bran:2005}
F.~Bretz, J.~C. Pinheiro, and M.~Branson.
\newblock Combining multiple comparisons and modeling techniques in
  dose-response studies.
\newblock \emph{Biometrics}, 61:\penalty0 738--748, 2005.

\bibitem[Chatfield(1995)]{chat:1995}
C.~Chatfield.
\newblock Model uncertainty, data mining and statistical inference.
\newblock \emph{Journal of the Royal Statistical Society, Series A},
  158:\penalty0 419--466, 1995.

\bibitem[Claeskens and Hjort(2008)]{clae:hjor:2008}
G.~Claeskens and N.~L. Hjort.
\newblock \emph{Model Selection and Model Averaging}.
\newblock Cambridge University Press, Cambridge, 2008.

\bibitem[Davies(1977)]{davi:1977}
R.~B. Davies.
\newblock Hypothesis testing when a nuisance parameter is present only under
  the alternative.
\newblock \emph{Biometrika}, 64:\penalty0 247--254, 1977.

\bibitem[Davies(1987)]{davi:1987}
R.~B. Davies.
\newblock Hypothesis testing when a nuisance parameter is present only under
  the alternative.
\newblock \emph{Biometrika}, 74:\penalty0 33--43, 1987.

\bibitem[Dette et~al.(2015)Dette, Titoff, Volgushev, and
  Bretz]{dett:tito:bret:2015}
H.~Dette, S.~Titoff, S.~Volgushev, and F.~Bretz.
\newblock Dose response signal detection under model uncertainty.
\newblock \emph{Biometrics}, 2015.
\newblock Epub ahead of print.

\bibitem[Diaconis and Efron(1985)]{diac:1985}
P.~Diaconis and B.~Efron.
\newblock Testing for independence in a two-way table: new interpretations of
  the chi-square statistic.
\newblock \emph{The Annals of Statistics}, 13:\penalty0 845--874, 1985.

\bibitem[Dragalin et~al.(2007)Dragalin, Hsuan, and
  Padmanabhan]{drag:hsua:padm:2007}
V.~Dragalin, F.~Hsuan, and S.~K. Padmanabhan.
\newblock Adaptive designs for dose-finding studies based on the sigmoid emax
  model.
\newblock \emph{Journal of Biopharmaceutical Statistics}, 17:\penalty0
  1051--1070, 2007.

\bibitem[{European Medicines Agency}(2014)]{ema:2014}
{European Medicines Agency}.
\newblock Qualification opinion of {MCP-Mod} as an efficient statistical
  methodology for model-based design and analysis of {P}hase {II} dose finding
  studies under model uncertainty, 2014.
\newblock http://goo.gl/imT7IT.

\bibitem[Fang and Zhang(1990)]{fang:1990}
K.~Fang and Y.~Zhang.
\newblock \emph{Generalized multivariate analysis}.
\newblock Science Press, 1990.

\bibitem[Fang and Wang(1994)]{fang:wang:1994}
K.-T. Fang and Y.~Wang.
\newblock \emph{Number-theoretic Methods in Statistics}.
\newblock Chapmann and Hall, London, 1994.

\bibitem[Gray(2004)]{gray:2004}
A.~Gray.
\newblock \emph{Tubes}.
\newblock Birkh{\"a}user, 2nd edition, 2004.

\bibitem[Grieve and Krams(2005)]{grie:kram:2005}
A.~P. Grieve and M.~Krams.
\newblock {ASTIN}: a {B}ayesian adaptive dose-response trial in acute stroke.
\newblock \emph{Clinical Trials}, 2:\penalty0 340--351, 2005.

\bibitem[Gupta and Varga(1993)]{gupt:1993}
A.~Gupta and T.~Varga.
\newblock \emph{Elliptically Contoured Models in Statistics}.
\newblock Springer, 1993.

\bibitem[Hotelling(1939)]{hote:1939}
H.~Hotelling.
\newblock Tubes and spheres in n-spaces, and a class of statistical problems.
\newblock \emph{American Journal of Mathematics}, 61:\penalty0 440--460, 1939.

\bibitem[Johansen and Johnstone(1990)]{joha:john:1990}
S.~Johansen and I.~M. Johnstone.
\newblock Hotelling's theorem on the volume of tubes: Some illustrations in
  simulatenous inference and data analysis.
\newblock \emph{Annals of Statistics}, 18:\penalty0 652--684, 1990.

\bibitem[Jones et~al.(2011)Jones, Layton, Richardson, and
  Thomas]{jone:layt:rich:2011}
B.~Jones, G.~Layton, H.~Richardson, and N.~Thomas.
\newblock Model-based {B}ayesian adaptive dose-finding designs for a phase {II}
  trial.
\newblock \emph{Statistics in Biopharmaceutical Research}, 3:\penalty0
  276--287, 2011.

\bibitem[Leeb and P\"{o}tscher(2005)]{leeb:poet:2005}
H.~Leeb and B.~M. P\"{o}tscher.
\newblock Model selection and inference: Facts and fiction.
\newblock \emph{Econometric Theory}, 21:\penalty0 21--59, 2005.

\bibitem[Lehmann and Romano(2008)]{lehm:roma:2008}
E.~Lehmann and J.~Romano.
\newblock \emph{Testing Statistical Hypotheses}.
\newblock Springer, New York, 3rd edition, 2008.

\bibitem[Li(2011)]{li:2011}
S.~Li.
\newblock Concise formulas for the area and volume of a hyperspherical cap.
\newblock \emph{Asian Journal of Mathematics and Statistics}, 4:\penalty0
  66--70, 2011.

\bibitem[Liu and Shao(2003)]{liu:shao:2003}
X.~Liu and Y.~Shao.
\newblock Asymptotics for likelihood ratio tests under loss of identifiability.
\newblock \emph{Annals of Statistics}, 31:\penalty0 807--832, 2003.

\bibitem[Marcus(1976)]{marc:1976}
R.~Marcus.
\newblock The power of some tests for the equality of normal means against an
  ordered alternative.
\newblock \emph{Biometrika}, 63:\penalty0 177--183, 1976.

\bibitem[Naiman(1990)]{naiman:1990}
D.~Naiman.
\newblock On volumes of tubular neighborhoods of spherical polyhedra and
  statistical inference.
\newblock \emph{The Annals of Statistics}, 18:\penalty0 685--716, 1990.

\bibitem[Pinheiro et~al.(2014)Pinheiro, Bornkamp, Glimm, and
  Bretz]{pinh:born:glim:2014}
J.~C. Pinheiro, B.~Bornkamp, E.~Glimm, and F.~Bretz.
\newblock Model-based dose finding under model uncertainty using general
  parametric models.
\newblock \emph{Statistics in Medicine}, 33:\penalty0 1646--1661, 2014.

\bibitem[Pukkila and Rao(1988)]{pukkila:1988}
T.~Pukkila and R.~Rao.
\newblock Pattern recognition based on scale invariant discriminant functions.
\newblock \emph{Information sciences}, 45:\penalty0 379--389, 1988.

\bibitem[Ritz and Skovgaard(2005)]{ritz:skov:2005}
C.~Ritz and I.~M. Skovgaard.
\newblock Likelihood ratio tests in curved exponential families with nuisance
  parameters present only under the alternative.
\newblock \emph{Biometrika}, 92:\penalty0 507--517, 2005.

\bibitem[Ross(2011)]{ross:2011}
S.~Ross.
\newblock \emph{A First Course in Probability}.
\newblock Prentice Hall, 8th edition, 2011.

\bibitem[Schorning et~al.(2015)Schorning, Bornkamp, Bretz, and
  Dette]{scho:born:bret:2015}
K.~Schorning, B.~Bornkamp, F.~Bretz, and H.~Dette.
\newblock Model selection versus model averaging in dose finding studies.
\newblock \emph{Technical Report}, 2015.
\newblock http://arxiv.org/pdf/1508.00281.pdf.

\bibitem[Seber and Lee(2003)]{seber:2003}
G.~Seber and A.~Lee.
\newblock \emph{Linear Regression Analysis}.
\newblock Wiley, 2nd edition, 2003.

\bibitem[Thomas(2006)]{thom:2006}
N.~Thomas.
\newblock Hypothesis testing and {B}ayesian estimation using a sigmoid {E}max
  model applied to sparse dose designs.
\newblock \emph{Journal of Biopharmaceutical Statistics}, 16:\penalty0
  657--677, 2006.

\bibitem[Weyl(1939)]{weyl:1939}
H.~Weyl.
\newblock On the volume of tubes.
\newblock \emph{American Journal of Mathematics}, 61:\penalty0 461--472, 1939.

\bibitem[Williams(1971)]{will:1971}
D.~A. Williams.
\newblock A test for difference between treatment means when several dose
  levels are compared to a zero dose control.
\newblock \emph{Biometrics}, 27:\penalty0 103--117, 1971.

\end{thebibliography}

\newpage

\appendix

\large{\textbf{Appendix A:}} \normalsize
\textbf{Likelihood-ratio statistic for a single parameter value}
\label{sec:likel-ratio-stat}

\begin{theorem}
\label{thm:stat}
The LR statistic
for $H_0 : \beta = 0$ against $H_i : \beta \ge  0$ in the linear model
$y \sim \mathcal{N}(\alpha \bone{n}+\beta x, \sigma^2 I_n)$
has the form $S(R) = (1 - \ind{R > 0} R^2)^{n/2}$,
where $R$ is defined in Equation (6) in the paper.
\end{theorem}

\begin{proof}
Note first that
$R =
(B x)^{\top}(B y) / (\norm{B x} \norm{B y}) =
(C x)^{\top}(C y) / (\norm{C x} \norm{C y})$,
with
$C = I_n - n^{-1}\bone{n}\bone{n}^{\top}$
the centering matrix.

Let us first deal with a trivial case:
\paragraph*{Case 1:}\ If $y$ is constant (that is, if
$y = \alpha \bone{n}$ for some $\alpha \in \mathbb{R}$),
then $R = 0$ (since $C y = \mathbf{0}_n$) and $S = 1$ (since
the likelihood is maximized for $\beta = 0$ both
under $H_0$ and under $H_i$).

Assume from now on that $y$ is non-constant.
Then the LR statistic has the form \citep[p.~99]{seber:2003}
$S = (1 - T)^{n/2}$ with
\[
T = 1 - \frac{\min_{\alpha\in\mathbb{R}, \beta\in\mathbb{R}_{+}}
\norm{y - \alpha\bone{n} - \beta x}^2}
{\min_{\alpha\in\mathbb{R}} \norm{y - \alpha\bone{n}}^2}
=
1 - \frac{f(\hat{\beta}_{+})}{f(0)},
\]
where
$f(\beta) = \norm{C(y - \beta x)}^2$
and
$\hat{\beta}_{+} = \argmin_{\beta\in\mathbb{R}_{+}} f(\beta)$
(with $f(0) > 0$ since $y$ is non-constant).

Define
$\hat{\beta}_{\pm} = \argmin_{\beta\in\mathbb{R}} f(\beta)$.
From $\hat{\beta}_{\pm} = (\norm{C y} / \norm {C x}) R$,
\citet[p.~139]{seber:2003}, it follows that
\begin{equation}
  \label{eq:betapm}
  \hat{\beta}_{\pm} > 0\quad \Longleftrightarrow\quad R > 0.
\end{equation}

Let us distinguish two cases:
\paragraph*{Case 2:}\ If $R \le 0$,
then $\hat{\beta}_{\pm} \le 0$ and $\hat{\beta}_{+} = 0$
(since $f$ is quadratic in $\beta$), so that
\begin{equation}
  \label{eq:case1}
  T = 1 - f(0) / f(0) = 0.
\end{equation}
\paragraph*{Case 3:}\ If $R > 0$, then $\hat{\beta}_{\pm} > 0$
and $\hat{\beta}_+ = \hat{\beta}_{\pm}$, so that
\[
  T = 1 - f(\hat{\beta}_{\pm}) / f(0)
\]
From
\[
\norm{Cy - (\norm{C y} / \norm{C x}) R (C x)}^2 =
\norm{C y}^2 -
2\, \underbrace{(\norm{C y} / \norm{C x}) R (C x)^{\top}(C y)}_{\norm{Cy}^2 R^2}
+
\underbrace{(\norm{C y}^2 / \norm{C x}^2) R^2 \norm{C x}^2}_{\norm{Cy}^2 R^2}
\]
we get
$f(\hat{\beta}_{\pm}) = \norm{C y}^2 (1 - R^2)$ and
\begin{equation}
  \label{eq:case2b}
  T = 1 - \frac{\norm{C y}^2 (1 - R^2)}{\norm{Cy}^2} = R^2.
\end{equation}

Putting Case~2 and Case~3 together, we obtain $T = \ind{R > 0} R^2$,
and combining all three cases, we get
\begin{equation}
  \label{eq:SR}
  S(R) = (1 - \ind{R > 0} R^2)^{n/2}.
\end{equation}
\end{proof}

\vspace{0.7cm}
\large{\textbf{Appendix B:}} \normalsize
\textbf{Counterexample: $\chi^2$ critical values}
\label{sec:counter}

Consider, the model
\[
x_{\gamma} = B^{+} \tilde{x}_{\gamma}, \qquad \gamma \in \Gamma = [0, 2\pi],
%\phi(\gamma, t \gamma,
%\dots, t^{d-1} \gamma),
\]
with $B^{+}$ the Moore-Penrose pseudoinverse of $B$, and
\[
\tilde{x}^{(k)}_{\gamma} = \cos \left(\ind{k > 1} \lambda^{k-2} \gamma\right)
\prod_{\ell = k}^{d} \sin \left(\lambda^{\ell - 1} \gamma\right),\quad k = 1, \dots, {d + 1},
\]
for some constant $\lambda \in \mathbb{N}$.
The standardized predictions, $\tilde{x}_{\gamma}$, are the
Cartesian coordinates of a point with spherical coordinates
$\gamma,\, \lambda \gamma,\, \dots, \lambda^{d-1} \gamma$.
Since
$\sup_{\psi \in [0, 2\pi]^d}\inf_{\gamma\in\Gamma}\norm{\psi -
(\gamma \bmod 2 \pi,\, \lambda \gamma \bmod 2 \pi,\, \dots, \lambda^{d - 1}
\gamma \bmod 2 \pi)^{\top}}
\to 0$, also
$\inf_{\tilde{y}\in\mathbb{S}}\sup_{\gamma\in\Gamma}\tilde{y}^{\top}
  \tilde{x}_{\gamma} \to 1$, as $\lambda \to \infty$.
Therefore, for any $q$, we can choose a constant $\lambda$
so that $R > (1 - \exp(-q / n))^{1/2}$ for all $\tilde{y}\in\mathbb{S}$,
which ensures that $-2 \log S > q$.

The plot in the right-hand side of Figure 2 in the paper
illustrates the construction when $d = 2$.
As $\gamma$ goes from 0 to $\pi$,
the spiral curve $\gamma \mapsto \tilde{x}_{\gamma}$
takes $\lambda$ rotations around the sphere.
For any $r < 1$, one can choose $\lambda$ large enough
to make the tubular neighborhood around $\mathbb{M}$ cover
the complete sphere.
Note that both of the models in Figure 2
have the same number of parameters but their complexities
(in terms of the possible predictions) are vastly different.

\vspace{0.7cm}
\large{\textbf{Appendix C:}} \normalsize
\textbf{Details on the numerical calculation of tubular volumes}
\label{app:num}

% This section describes a Monte Carlo algorithm
% to approximate the integral
% and $f : \mathbb{T}_r \to \mathbb{R}$ a bounded Borel function.
% Special cases are the volume
%  $\abs{\mathbb{T}_r}$
% when $f = \abs{\mathbb{S}}$ (due to the area formula,
% \cite[Section~3.2.5]{fede:1969}),
% the p-value $P_0(R > r)$
% when $f = 1$,
% and the power
% $P_1(R > r)$
% when $f$ is an angular Gaussian density
% under some alternative hypothesis.

Assume that we have some way to sample a point $W$ with support
$\mathbb{M}$.
%by some unspecified routine.
Let us say that this point has distribution $H$ on $\mathbb{M}$, even
if we do not know $H$ explicitly.  For the sake of illustration, here
is one simple way to sample such a point: start by selecting a model
$i$ from the $m$ models with equal probability, then sample a value
$\gamma$ from $\Gamma_i$ according to some convenient distribution and
finally take $W = \tilde{x}_{\gamma, i}$.

Next, we sample a point $V$ from $G_W$, the uniform probability
distribution on the spherical cap~$\mathbb{C}_{W r}$.  The joint
distribution is $P(V \in T, W \in M) = \int_M G_w(T) \mathrm{d} H(w)$,
for Borel sets $T \subseteq \mathbb{T}_{r}$ and $M \subseteq
\mathbb{M}$.  By Fubini's theorem, $\int_M G_w(T)\, \mathrm{d} H(w) =
\int_T \int_M g_w(v)\, \mathrm{d} H(w)\, \mathrm{d} \varsigma(v)$,
with $g_w(v) = \ind{w \in \mathbb{C}_{v r}} / \varsigma(\mathbb{C}_{v
  r})$ the density of $G_w$ with respect to $\varsigma$.  Let us
define $c_r = \varsigma(\mathbb{C}_{v r}) =
\abs{\mathbb{C}_{r}} / \abs{\mathbb{S}}$. The value of $c_r$
is given in Equation (7) in the main text. Then, the random
point~$V$ has density
\[g(v) =
\frac{1}{c_r}\int_{\mathbb{M}} \ind{w \in \mathbb{C}_{v r}}\, \mathrm{d}
H(w) = \frac{1}{c_r}P(W \in \mathbb{C}_{v r}).\]

If we repeat this procedure $\kappa$ times, we obtain a sample $(V_1,
W_1), \dots, (V_{\kappa}, W_{\kappa})$.  Approximating $P(W \in
\mathbb{C}_{v r})$ by relative frequencies gives $g(v) \approx
\abs{\{j : W_{j} \in \mathbb{C}_{v r}\}} / (\kappa c_r)$.  In this
way, we arrive at the approximation
\begin{equation}
\label{eq:approx}
\int_{\mathbb{T}_r} f\, \mathrm{d} \varsigma \approx
c_r
\sum_{k = 1}^{\kappa} \frac{f(V_k)}{\abs{\{j:W_j \in \mathbb{C}_{V_k r}\}}}.
\end{equation}

Note that we work with points on the $d$-dimensional unit sphere,
which becomes difficult when the sample size becomes large.  It is,
however, possible to modify the algorithm, and work with sufficient
statistics instead of raw data.  To keep the presentation brief, the
details are not given here. Algorithm~1 outlines the main steps of the
procedure.

\begin{algorithm}
  \caption{Sampling algorithm}
%  \SetAlgoLined
  \For {$k \leftarrow 1$ \KwTo $\kappa$}{
    sample $W_k$ from the distribution $H$ on $\mathbb{M}$\;
    sample $V_k$ from $G_{W_k}$ on $\mathbb{C}_{W_k,r}$\;
  }
  \For {$k \leftarrow 1$ \KwTo $\kappa$}{
  Calculate $m_k=\abs{\{j:W_j \in  \mathbb{C}_{V_k r}\}}$\;
  }
  \KwResult{$c_r \sum_{k = 1}^{\kappa} f(V_k) / m_k$}
\end{algorithm}

Theorem~\ref{thm:cons-sampl-scheme} below shows that the right-hand side of
Equation~\eqref{eq:approx} converges in probability to the left-hand
side as $\kappa \to \infty$.  Hence, the approximation error can be
made arbitrarily small by increasing the number of sampling replicates
$\kappa$.

\begin{definition}
In the following, we will call a family $Y_{n j}$, with
$j = 1, \dots, n$ and $n = 1, 2, \dots$, of
random vectors a \emph{triangular arrangement},
if for each $n$, the random vectors
$Y_{n1}, \dots, Y_{n n}$ are iid and have
finite expectations and variances.
\end{definition}

\newcommand{\intT}{\int_{\mathbb{T}_r} f\,\mathrm{d}\varsigma}

\begin{theorem}
\label{thm:cons-sampl-scheme}
Let $\mathbb{M}$ be a Borel set
on $\mathbb{S}$ and let
$\mathbb{T}_r$ denote the tubular neighborhood
around~$\mathbb{M}$ with radius $r$.
Consider a sequence
$w_n$ in $\mathbb{M}$
so that $\abs{\{j: 1\le j\le n\ \text{and}\ w_j \in \mathbb{C}_{v r}\}}
\to \infty$ as $n \to \infty$ for each $v$ in the interior of $\mathbb{T}_r$.
Also consider
a sequence of independent random vectors
$V_n \sim \mathcal{U}(\mathbb{C}_{w_n, r})$,
uniformly distributed on the spherical cap
$\mathbb{C}_{w_n, r}$.
Let $f : \mathbb{S} \to \mathbb{R}_+$
be Borel and bounded.
Define the sequence
of functions
$g_n : \mathbb{T}_{r n} \to \mathbb{R}_+$,
on $\mathbb{T}_{r n} = \cup_{j = 1}^n \mathbb{C}_{w_j, r}$,
by
$g_n(v) = (n\, \psi_n(v)\varsigma(\mathbb{C}_{v r}))^{-1}$
with
$\psi_n(v) = 1 / \abs{\{j: 1\le j\le n\ \text{and}\ w_j \in
  \mathbb{C}_{v r}\}}$ and
$\varsigma$ the uniform probability measure on $\mathbb{S}$.
Also define the random variables $X_n = f(V_n) / g_n(V_n)$ and
$S_n = (X_1 + \dots + X_n) / n$.
Then  $S_n \stackrel{P}{\to} \intT$.
\end{theorem}
\begin{proof}
  Consider the triangular arrangements
  $K_{n j}$ and $V_{n j}'$ and $X_{n j}'$,
  where $K_{n j} \sim \mathcal{U}(\{1, \dots, n\})$,
  where $V_{n j}' \given (K_{n j} = k) \sim
  \mathcal{U}(\mathbb{C}_{w_k, r})$,
  and where $X_{n j}' = f(V_{n j}') / g_n(V_{n j}')$.
  Define $S_n' = (X_{n 1}' + \dots + X_{n n}') / n$.
  Due to Lemma~\ref{lemma:mixture} below, it suffices to show that
  $E(S_{n}') \to \intT$ and $\var(S_n') \to 0$.

  First, the argument for $E(S_{n}') \to \intT$.
  For each $n$, the variables $X_{n1}', \dots X_{n n}'$ are
  iid; hence
  $E(S_n') = E(X_{n1}')$. Since $V_{n1}'$ has
  density $g_n$ on $\mathbb{T}_{r n}$ (with respect to $\varsigma$),
  \[
  E(X_{n 1}') = E\left(\frac{f(V_{n 1}')}{g_n(V_{n1}')}\right) =
  \int_{\mathbb{T}_{r n}} \frac{f(v)} {g_n(v)} g_n(v) \mathrm{d} \varsigma(v)
  = \int_{\mathbb{T}_{r n}} f\,\mathrm{d}\varsigma.
  \]
  Now $\mathbb{T}_{r,1} \subseteq
  \mathbb{T}_{r,2} \subseteq \dots$ is an increasing sequence
  of subsets of $\mathbb{T}_r$,
  and $\abs{\mathbb{T}_r \setminus \mathbb{T}_{r n}} \to 0$
  by the assumptions on the sequence $w_n$;
  therefore, by the monotone convergence theorem,
  $E(S_{n}') = \int_{\mathbb{T}_r} \ind{v \in \mathbb{T}_{r n}}
  f(v) \,\mathrm{d} \varsigma(v) \to \intT$.

  Second, the argument for $\var(S_n') \to 0$.
  Integrating with respect to the density $g_n$ of $V_{n 1}'$ gives
  \[
  \var(X_{n 1}') = E\bigl((X_{n 1}')^2\bigr) - \bigl(E(X_{n 1}')\bigr)^2
  \le \int_{\mathbb{T}_{r n}} \frac{f(v)^2}{g_n(v)}\, \mathrm{d} \varsigma(v).
  \]
  Let $M = \sup_{v \in \mathbb{T}_{r n}} f(v)^2 \varsigma(\mathbb{C}_{v r})$,
  which is finite since $f$ is bounded.
  Then $\var(X_{n 1}') \le n M \int_{\mathbb{T}_{r n}} \psi_n\,\mathrm{d}\varsigma$.
  Furthermore, $\var(S_n') = \var(X_{n1}') / n$
  since $X_{n1}', \dots X_{n n}'$ are iid;
  thus $\var(S_{n}') \le M \int_{\mathbb{T}_{r n}} \psi_n\,\mathrm{d}\varsigma$.
  But $\psi_n(v) \to 0$ for every $v$
  in the interior of $\mathbb{T}_r$ and the
  boundary of $\mathbb{T}_r$ has measure zero,
  so that $\var(S_n') \to 0$ by the dominated convergence theorem.
\end{proof}

\begin{lemma}
\label{lemma:mixture}
  Consider a sequence of independent random
  variables $X_n$ and
  two triangular arrangements
  $K_{n j}$ and $X_{n j}'$,
  where $K_{n j} \sim \mathcal{U}(\{1, \dots, n\})$ and
  $X_{n j}' \given (K_{n j} = k) \stackrel{d}{=} X_k$.
  Define $S_n = (X_1 + \dots + X_n) /n$ and $S_n' = (X_{n_1}' + \dots +
  X_{n n}') / n$.
  If $E(S_n') \to \alpha$
  for some value $\alpha \in \mathbb{R}$
  and $\var(S_n') \to 0$, then
  $S_n \stackrel{P}{\to} \alpha$.
\end{lemma}
\begin{proof}
  Fix a value $\epsilon > 0$. For every $n \in \mathbb{N}$,
  \[E(S_n) = \frac{1}{n} \sum_{j = 1}^{n} E(X_j) = E(E(X_{n1}'\given
  K_{n 1})) = E(X_{n1}') = E(S_n').
  \]
  Hence there exists an $m \in \mathbb{N}$ so that
  $\abs{E(S_n) - \alpha} < \epsilon / 2$, for all $n > m$.

  It follows that for such $n > m$ also
  $P(\abs{S_n - \alpha} > \epsilon) \le P(\abs{S_n - E(S_n)} > \epsilon / 2)$.
  By Chebyshev's inequality
  $P(\abs{S_n - E(S_n)} > \epsilon / 2)
  \le 4 \var(S_n) / \epsilon^2$. Consequently,
  if $\var(S_n) \to 0$, then
  $P(\abs{S_n - \alpha} > \epsilon) \to 0$.
  And $\var(S_n) \to 0$ follows from
  \begin{equation*}
    \begin{split}
      \var(S_n) & = \frac{1}{n^2}\sum_{k =1}^n \var(X_k) \\
      & = \frac{1}{n} E(\var(X_{n1}'\given K_{n1})) \\
      & \le \frac{1}{n} \Bigl(E(\var(X_{n1}'\given K_{n1}))
        + \var(E(X_{n j}'\given K_{n1}))\Bigr) \\
      & = \frac{1}{n} \var(X_{n1}') = \var(S_n'),
    \end{split}
  \end{equation*}
where the last line uses the well known equation for the conditional
variance (see, e.g., Proposition 5.2 in \cite{ross:2011}).
\end{proof}

\end{document}